\documentclass[11pt]{amsart}
\newcommand{\ignore}[1]{}
\usepackage{a4wide}
\usepackage{lmodern}
\usepackage{stix}
\usepackage{amssymb}
\usepackage{amsfonts}
\usepackage{microtype}
\usepackage{amsmath}
\usepackage{booktabs} 
\usepackage{amssymb}
\usepackage{amssymb}
\usepackage{amsthm}
\usepackage{bbm} 
\usepackage{textcomp}
\usepackage{xcolor}

\usepackage[colorlinks]{hyperref}
\definecolor{gruen}{rgb}{0,.5,0}
\hypersetup{colorlinks=true,linkcolor=gruen,citecolor=gruen,urlcolor=gruen}

\ignore{
\usepackage{crossreftools}
\pdfstringdefDisableCommands{%
    \let\Cref\crtCref
    \let\cref\crtcref}
\hypersetup{bookmarksnumbered, bookmarksopen=true, bookmarksopenlevel=1}
}

\usepackage[capitalise]{cleveref}
\crefname{lem}{Lemma}{Lemmas}
\crefname{clm}{Claim}{Claims}
\crefname{thm}{Theorem}{Theorems}

\newtheorem{thm}{Theorem}
\newtheorem{lem}{Lemma}
\newtheorem{cor}{Corollary}
\newtheorem{clm}{Claim}

\theoremstyle{remark}
\newtheorem{ex}{Example}
\newtheorem{rem}{Remark}

\let\cal\mathcal
\newcommand{\bigO}{\ensuremath{\mathcal{O}}}
\usepackage{upgreek} 
\newcommand{\NN}{\mathbb{N}} 
\newcommand{\RR}{\mathbb{R}}
\newcommand{\QQ}{\mathbb{Q}}
\newcommand{\ZZ}{\mathbb{Z}}
\newcommand{\F}{F}
\newcommand{\minplus}{$(\min,+)$~}
\newcommand{\maxplus}{$(\max,+)$~}
\newcommand{\prob}[1]{\ensuremath{\mathrm{Pr}\left\{#1\right\}}}
\newcommand{\pprob}{\ensuremath{\mathrm{Pr}}}
\newcommand{\euler}{\mathrm{e}}

\newcommand{\h}{{\cal H}}
\newcommand{\basis}{{\cal B}}
\newcommand{\cBPP}{\mathsf{BPP}}
\newcommand{\cP}{\mathsf{P}}
\newcommand{\cPoly}{\mathsf{P}\kern-0.1em/\kern-0.1em\mathrm{poly}}
\newcommand{\cPolyBPP}{\mathsf{BPP}\kern-0.1em/\kern-0.1em\mbox{$\mathsf{poly}$}}
\newcommand{\skal}[1]{\langle #1\rangle}
\newcommand{\pred}[1]{\left[#1\right]}
\DeclareMathOperator{\maj}{Maj}
\DeclareMathOperator{\mmaj}{maj}
\DeclareMathOperator{\sort}{sort}
\DeclareMathOperator{\ssel}{Sel}
\newcommand{\sel}[2]{\ssel(#1|#2)}
\newcommand{\supp}[1]{S_{\!#1}}
\newcommand{\xx}{\boldsymbol{x}}
\newcommand{\xr}{\boldsymbol{r}}
\newcommand{\xf}{\boldsymbol{F}}
\newcommand{\xh}{\boldsymbol{H}}
\newcommand{\rr}{r}
\newcommand{\size}{s}

\newcommand{\err}{\epsilon}
\newcommand{\sgn}[1]{\mathrm{sgn}\,{#1}}
\newcommand{\ssgn}{\mathrm{sgn}}
\newcommand{\rel}{\, \varrho\, } 
\newcommand{\rrel}{\varrho}
\newcommand{\Maj}{\mu}

\newcommand{\Diam}{\diamondsuit}
\newcommand{\factor}{c}

\newcommand{\ef}[2]{\mathrm{ave}_{#1}(#2)}
\newcommand{\tp}[1]{p_{#1}}
\newcommand{\symm}[2]{P_{#1}(#2)}
\newcommand{\ssymm}[1]{P_{#1}}
\newcommand{\rv}{\xi}
\newcommand{\indf}{f} 
\newcommand{\slocal}{t}
\newcommand{\K}{K}
\newcommand{\sslocal}[1]{\slocal_{#1}}
\newcommand{\dgates}{b}
\newcommand{\relcompl}{\slocal_{\rrel}}
\newcommand{\dom}{D}
\newcommand{\range}{R}
\newcommand{\const}{c}

\newcommand{\dist}[1]{\pi(#1)}

\def\tabelle{ \setlength{\tabcolsep}{5pt}
  \begin{table}[t]
    \caption{\footnotesize Examples of semialgebraic functions, where
      $m$ is the number of distinct polynomials used in a formula, and $d$ is their degree, that is, the largest sum of degrees of variables appearing in any monomial of these polynomials.
      Here,  $p(x)$ is an arbitrary real multivariate polynomial of degree $d$, and $\Psi(x)$ is a  algebraic formula using $s$ polynomials of maximum degree $d\geq 1$; $\sel{x_1,\ldots,x_n}{y}$ is a partly defined function that outputs $x_i$ if $y=i$. In the algebraic formulas for the majority
      vote functions,  $\mmaj$ is the Boolean majority function.
      \medskip}
    \label{tab:basic}
    \begin{tabular}{lcc}
        \toprule
        Graph of $f$ &  $(m,d)$ & Algebraic formula $\Phi$\\
        \midrule
        $z=p(x)$  & $(1,d)$ & $\pred{z=p(x)}$\\[1ex]
        $z=|x|$ &  $(3,1)$  & $(\pred{x\geq 0}\land \pred{z=x})\lor(\pred{x< 0}\land \pred{z=-x})$\\[1ex]
        $z=x^{1/k}$ & $(2,k)$ &  $\pred{x=z^k}$ (odd $k$),  $\pred{x\geq 0}\land\pred{x=z^k}$ (even $k$)\\[1ex]
        $z=\|x-y\|$ &  $(2,2)$  & $\pred{z\geq 0}\land \pred{z^2=(x_1-y_1)^2+\cdots+(x_n-y_n)^2}$\\[1ex]
        $z=x/y$ & $(2,2)$ & $\pred{y\neq 0}\land \pred{y\cdot z=x}$\\[1ex]
        $z=\min(x,y)$ & $(2,1)$ & $\pred{z\leq x}\land \pred{z\leq y}\land( \pred{z=x}\lor\pred{z=y})$\\[1ex]
        $z=\max(x,y)$ & $(2,1)$ & $\pred{z\geq x}\land \pred{z\geq y}\land( \pred{z=x}\lor\pred{z=y})$\\[1ex]
        $z=\maj(x_1,\ldots,x_n)$ & $(n,1)$ & $\mmaj\big(\pred{z=x_1},\ldots,\pred{z=x_n}\big)$\\[1ex]
        $z=\sel{x_1,\ldots,x_n}{y}$ & $(2n,1)$ & $\bigvee_{i=1}^n\pred{y=i}\land\pred{z=x_i}$  \\[1ex]
        $z= \mbox{``if $\Psi(x)=1$ then $u$ else $v$''}$ & $(s+2,d)$ & $(\Psi(x)\land\pred{z=u})\lor (\neg \Psi(x)\land\pred{z=v})$\\
        \bottomrule
      \end{tabular}
  \end{table}
}

\begin{document}


\title[]{Coin Flipping in Dynamic
  Programming is Almost Useless}

\author[]{Stasys Jukna$^{\ast}$}
\thanks{$^{\ast}$Faculty of Mathematics and Computer Science, Vilnius
    University, Lithuania}
    \thanks{$^{\phantom{*}}$Email: stjukna@gmail.com,\ homepage: \href{http://www.thi.cs.uni-frankfurt.de/~jukna/}
    {http://www.thi.cs.uni-frankfurt.de/$\sim$jukna/}}

  \begin{abstract}
    We consider probabilistic circuits working over the real numbers,
    and using arbitrary semialgebraic functions of bounded description
    complexity as gates.  In particular, such circuits can use all
    arithmetic operations $+,-,\times,\div$, optimization operations
    $\min$ and $\max$, conditional branching (if-then-else), and many
    more.  We show that probabilistic circuits using any of these
    operations as gates can be simulated by deterministic circuits
    with only about a quadratical blowup in size.  A not much larger
    blow up in circuit size is also shown when derandomizing
    approximating circuits. The algorithmic consequence, motivating
    the title, is that randomness cannot substantially speed up
    dynamic programming algorithms.
  \end{abstract}

\maketitle


\section{Introduction}

Probabilistic algorithms can make random choices during their
execution. Often, such algorithms are more efficient than \emph{known}
deterministic solutions; see, for example, the
books~\cite{MotwaniR95,MitzenmacherU05}. So, a natural questions
arises: is randomness a really useful resource, can randomization
indeed substantially speed up algorithms? In the computational
complexity literature, this is the widely open\footnote{$\cBPP$ stands
  for ``{\bf b}ounded-error {\bf p}robabilistic {\bf p}olynomial
  time,'' and $\cP$ for ``deterministic {\bf p}olynomial time.''}
``$\cBPP$ versus $\cP$'' question. The nonuniform version of this
question, known as the ``$\cBPP$ versus $\cPoly$,'' question asks
whether probabilistic \emph{circuits} can be efficiently simulated by
deterministic circuits.

A \emph{probabilistic} circuit is a deterministic circuit that is
allowed to use additional input variables, each being a \emph{random
  variable} taking its values in the underlying domain. We allow
arbitrary probability distributions of these random variables, so that
our derandomization results will be distribution independent. Such a
circuit \emph{computes} a given function $f$ if, on every input $x$,
the circuit outputs the correct value $f(x)$ with probability at
least\footnote{There is nothing ``magical'' in the choice of this
  threshold value $2/3$: we do this only for definiteness. One can
  take any constant \emph{larger} than $1/2$: since we ignore
  multiplicative constants in our bounds, all results will hold also
  then.}~$2/3$. The \emph{size} of a (deterministic or probabilistic)
circuit is the number of used gates.

A classical result of Adleman~\cite{adleman}, extended to the case of
two-sided error probability by Bennett and Gill~\cite{bennett}, has
shown that randomness is useless in \emph{Boolean} circuits: if a
Boolean function $f$ of $n$ variables can be computed by a
probabilistic Boolean circuit of size polynomial in $n$, then $f$ can
be also computed by a deterministic Boolean circuit of size polynomial
in $n$. So, $\cBPP\subseteq \cPoly$ holds for Boolean circuits.

In this paper, we are mainly interested in the $\cBPP$ versus $\cPoly$
question for \emph{dynamic programming} algorithms (DP algorithms):
\begin{itemize}
\item {\it Can randomization substantially speed up DP algorithms?}
\end{itemize}
We answer this question in the \emph{negative}: randomized DP
algorithms \emph{can} be derandomized. That is, $\cBPP\subseteq
\cPoly$ holds also for DP algorithms. In fact, we prove a much
stronger result: $\cBPP\subseteq \cPoly$ holds for circuits over
\emph{any} basis consisting of semialgebraic operations
$g:\RR^l\to\RR$ of bounded algebraic description complexity.  We will
also show that the inclusion $\cBPP\subseteq \cPoly$ holds even when
circuits are only required to \emph{approximate} the values of given
functions.

Proofs of $\cBPP\subseteq \cPoly$ for \emph{Boolean} circuits
in~\cite{adleman,bennett} crucially used the fact that the domain
$\{0,1\}$ of such circuits is \emph{finite}: the proof is then
obtained by a simple application of the union and Chernoff's bounds
(see \cref{clm:fin-maj} in \cref{sec:maj}).  A trivial reason why such
a simple argument cannot derandomize DP algorithms is that these
algorithms work over \emph{infinite} domains such as $\NN$, $\ZZ$,
$\QQ$ or $\RR$ (inputs for optimization problems), so that already the
union bound badly fails.

One also faces the ``infinite domain'' issue, say, in the polynomial
identity testing problem over infinite fields; see, for example,
surveys~\cite{saxena,ShpilkaY10}. But when derandomizing DP
algorithms, we additionally face the ``non-arithmetic basis'' issue:
besides arithmetic $+,-\times,\div$ operations, such circuits can use
additional non-arithmetic operations, like tropical $\min$ and $\max$
operations, sorting, conditional branching (if-then-else), argmin,
argmax, and other complicated operations.

To nail all this (infinite domain and powerful gates), in this paper,
we consider the derando\-mi\-za\-tion of circuits that can use
\emph{any} semialgebraic functions of bounded description complexity
as gates.

A function $f:\RR^n\to\RR$ is \emph{semialgebraic} if its graph can be
obtained by finitely many unions and intersections of sets defined by
a polynomial equality or strict inequality. The \emph{description
  complexity} of $f$ is the minimum number $\slocal$ for which such a
representation of the graph of $f$ is possible by using at most
$\slocal$ distinct polynomials, each of degree at most $\slocal$
(see~\cref{sec:semialg} for precise definitions). All operations
mentioned above are semialgebraic of small description complexity;
see~\cref{tab:basic} in~\cref{sec:semialg} for more examples.

\subsubsection*{Derandomization of exactly computing circuits}
The \emph{majority vote} function is a partly defined function
$\maj(x_1,\ldots,x_m)$ which outputs the majority element of its input
string, if there is one. That is,
\begin{equation}\label{eq:maj}
  \mbox{$\maj(x_1,\ldots,x_m)=y$ if $y$
    occurs $>m/2$ times among the $x_1,\ldots,x_m$.}
\end{equation}
For example, in the case of $m=5$ variables, we have
$\maj(a,b,c,b,b)=b$, whereas the value of $\maj(a,b,c,a,b)$ is
undefined. The description complexity of $\maj(x_1,\ldots,x_m)$ is at
most $m$; see \cref{tab:basic} in \cref{sec:semialg}.

A \emph{deterministic copy} of a probabilistic circuit is a
deterministic circuit obtained by fixing the values of its random
input variables.  A (deterministic or probabilistic) circuit is
$\dgates$-\emph{semialgebraic} if each its basis operation (a gate)
has description complexity at most $\dgates$.  Note that $\dgates$
here is a \emph{local} parameter: it bounds the description complexity
of only \emph{individual} gates, not of the entire function computed
by the circuit. For example, circuits using any of the gates
$+,-,\times,\div$, $\min$, $\max$, ``if $x<y$ then $u$ else $v$'' are
$\dgates$-semialgebraic for $\dgates\leq 3$.

\begin{thm}\label{thm:main1}
  If a function $f:\RR^n\to\RR$ can be computed by a probabilistic
  $\dgates$-semialgebraic circuit of size $\size$, then $f$ can be
  also computed as a majority vote of $m=\bigO(n^2\size\log
  \dgates\size)$ deterministic copies of this circuit.
\end{thm}

Note that, even though the majority vote functions are only
\emph{partially} defined, the derandomized circuit ensures that, on
every input $x\in\RR^n$ to the circuit, the sequence of values given
to the last majority vote gate will \emph{always} (for every input $x$
to the entire circuit) contain a majority element.

Note also that the upper bound on the number $m$ of deterministic
copies in the derandomized circuit only depends on the number $n$ of
deterministic input variables, on the number $\size$ of gates in the
probabilistic circuit, and on the (logarithm of) the description
complexity $\dgates$ of individual gates.  But it depends neither on
the fanin of gates, nor on the number of random input variables.

\subsubsection*{Derandomization of approximating circuits}
Our next (and main) result derandomizes probabilistic circuits when
they are only required to \emph{approximate} the values of a given
function (instead of computing the function exactly, as in
\cref{thm:main1}).

Let $x\rel y$ be any binary relation between real numbers $x,y\in\RR$.
One may interpret $x\rel y$ (especially, in the context of
approximating algorithms) as ``$x$ lies close to $y$.''  The
description complexity of the relation $\rel$ is the description
complexity of the set $S=\{(x,y)\in\RR^2\colon x\rel y\}$.

Given a binary relation $x\rrel y$ between real numbers, we say that a
probabilistic circuit $\F(x,\xr)$ $\rrel$-\emph{approximates} a given
function $f(x)$ if, for every input $x\in\RR^n$, $\F(x,\xr)\rel f(x)$
holds with probability at least $2/3$. That is, on every input $x$,
the circuit only has to output a value which is ``close enough'' to
the correct value $f(x)$ with probability at least~$2/3$.

\begin{ex}\label{ex:relations}
  Some of the most basic relations are the following ones.
  \begin{enumerate}
  \item Equality relation: $x\rel y$ iff $x=y$.
  \item Sign relation: $x\rel y$ iff $x=y=0$ or $x\cdot y > 0$.
  \item Nullity relation: $x\rel y$ iff $x=y=0$ or $x\cdot y\neq 0$.
  \item Approximation relation: $x\rel y$ iff $|x-y|\leq \factor$ for
    some fixed number $\factor\geq 0$.
  \end{enumerate}
  In the case of approximating circuits, the first relation (1)
  corresponds to computing the values $f(x)$ exactly, as in
  \cref{thm:main1}. The second relation (2) corresponds to detecting
  signs of the values $f(x)$.  In the case of relation (3), a circuit
  must recognize the roots of $f$, that is, must output $0$ precisely
  when $f(x)=0$. In the case of the last relation (4), the values
  computed by the circuit must lie not far away from the correct
  values~$f(x)$.
\end{ex}

A \emph{majority $\rrel$-vote function} is a (partial) function
$\Maj:\RR^m\to\RR$ with the following property for any real numbers
$a,x_1,\ldots,x_m$:
\begin{equation}\label{eq:maj-rel}
  \mbox{if $x_i\rel a$ holds for more than $m/2$ positions $i$, then
    $\Maj(x_1,\ldots,x_m)\rel a$ holds.}
\end{equation}
That is, if more than half of the input numbers $x_1,\ldots,x_m$ lie
close to the number $a$, then also the value of $\Maj$ must lie close
to~$a$.  For example, the majority vote function $\maj$ is the unique
majority $\rrel$-vote function for the equality relation (when $x\rel
y$ iff $x=y$). In general, however, there may be more than one
majority $\rrel$-vote function. For example, for a function
$\Maj:\RR^m\to\RR$ to be a majority $\rrel$-vote function for the
nullity relation $\rel$ it is enough that $\Maj(x)=0$ if more than
half of input numbers are zeros, and $\Maj(x)\neq 0$ otherwise.

In the following theorem, $x\rel y$ is an arbitrary
$\relcompl$-semialgebraic relation, and $f:\RR^n\to\RR$ a
$\sslocal{f}$-semialgebraic function.

\begin{thm}[Main result]
  \label{thm:main2}
  If $f$ can be $\rel$-approximated by a probabilistic
  $\dgates$-semialgebraic circuit of size $\size$, then $f$ can be
  also $\rrel$-approximated as a majority $\rrel$-vote of
  $m=\bigO(n^2\size\log \K)$ deterministic copies of this circuit,
  where $\K=\dgates\size+\sslocal{f} +\relcompl$.
\end{thm}

Note that now (unlike in \cref{thm:main1}) the size of the
derandomized circuit depends (albeit only logarithmically) on the
description complexities $\sslocal{f}$ and $\relcompl$ of the function
$f$ approximated and of the approximation relation $\rrel$.  Although
$\sslocal{f}$ may be large, the description complexity of the
approximation relations $\rrel$ is usually small; say, for all four
relations mentioned in \cref{ex:relations} we have $\relcompl \leq 2$.

\subsubsection*{The majority vote ``issue''}
One issue still remains: just like in \cref{thm:main1}, the
deterministic circuits given by \cref{thm:main2} are not in a ``pure''
form: they require one additional majority $\rrel$-vote operation to
output their values.  To obtain a ``pure'' circuit, we have to compute
this operation by a (possibly small) circuit using only basis
operations.

In some weak bases, even the (standard) majority vote function
\eqref{eq:maj} (for the equality relation) cannot be computed at all. For example, arithmetic
$(+,-,\times)$ circuits, as well as tropical \minplus and \maxplus
circuits cannot compute majority vote functions (\cref{clm:maj-impos}
in \cref{app:not-maj}).  In most bases, however, majority vote
functions are easy to compute, even for general approximation
relations $\rrel$, not just for the equality relation.

Call a relation $x\rel y$ \emph{contiguous} if $x\leq y\leq z$, $x\rel
a$ and $z\rel a$ imply $y\rel a$. That is, if the endpoints of an
interval are close to $a$, then also all numbers in the interval are
close to~$a$.  Note that the relations (1), (2) and (4) mentioned in
\cref{ex:relations} are contiguous. It can be easily shown (see
\cref{clm:contiguous} in \cref{app:not-maj}) that:
\begin{itemize}
\item For every contiguous relation $x\rel y$, a majority $\rrel$-vote
  function of $m$ variables can be computed by a fanin-$2$
  $(\min,\max)$ circuit of size $\bigO(m\log m)$ .
\end{itemize}
The nullity relation is \emph{not} contiguous: take, for example,
$x=-1, y=0$ and $z=a=1$.  Then $x\leq y\leq z$, $x\rel a$ and $z\rel
a$ hold but $y\rel a$ does not hold: $y=0$ but $a\neq 0$. Still,
majority $\rrel$-vote function of $m$ variables for the nullity
relation can be also computed by a $(\min,\max,\times)$ circuit using
$\bigO(m\log m)$ gates, and by a monotone arithmetic $(+,\times)$
circuit using $\bigO(m^2)$ gates (\cref{clm:nullity} in
\cref{app:not-maj}).

Thus, if the approximation relation is contiguous, and if the
operations $(+,\times)$ or the operations $(\min,\max,\times)$ are
available as gates, then \cref{thm:main2} gives a ``pure'' circuit
(without a majority vote gate) of size $\bigO(ms+m\log
m)=\bigO(n^2\size^2\log\K)$, that is:
\begin{itemize}
\item The blow up in the size of the derandomized circuit is only
  about quadratic.
\end{itemize}

\begin{rem}[Relation to dynamic programming]
  Most (if not all) DP algorithms in discrete optimization use only
  several semialgebraic functions of small description complexity in
  their recursion equations: min, max, arithmetic operations, and
  apparently some additional, but still semialgebraic operations of
  small description complexity, like the selection or the
  ``if-then--else'' operations (see \cref{tab:basic} in
  \cref{sec:semialg}). So, \cref{thm:main1} implies that randomization
  is (almost) \emph{useless} in DP algorithms, at least as long as we
  are allowed to use \emph{different} deterministic DP algorithms to
  solve optimization problems on inputs $x\in\RR^n$ from
  \emph{different} dimensions~$n$. In fact, the message of this paper
  is even stronger: \cref{thm:main2} shows that randomization is
  almost useless also for \emph{approximating} DP algorithms.
\end{rem}

\begin{rem}[The ``uniformity'' issue]\label{rem:uniform}
  Usually, a DP algorithm is described by giving \emph{one} collection
  of recursion equations that can be applied to inputs of \emph{any}
  dimension $n$. In this respect, DP algorithms are ``uniform'' (like
  Turing machines). Probabilistic DP algorithms may use random input
  weights in their recursion equations. However, when derandomizing
  such algorithms, we do not obtain also \emph{one} collection of
  recursion equations valid for inputs of \emph{all} dimensions. What
  we obtain is a \emph{sequence} of deterministic DP algorithms, one
  for each dimension~$n$. To our best knowledge, in the ``uniform''
  setting (with $ \mathsf{P}$ instead of $\cPoly$), the inclusion
  $\cBPP\subseteq \mathsf{P}$ remains \emph{not} known to hold for DP
  algorithms, and even for ``pure'' DP algorithms using only \minplus
  or \maxplus operations in their recursion equations.
\end{rem}

\subsubsection*{Organization} \Cref{sec:previous} shortly summarizes
previous work towards derandomization of probabilistic decision trees and circuits
working over infinite domains. In \cref{sec:semialg}, we recall the
notions of semialgebraic functions and probabilistic
circuits. \Cref{sec:enroute} describes the three steps in which we
will come from probabilistic to deterministic circuits.  The next
three sections (\cref{sec:growth,sec:quantifiers}) contain technical
results used to implement these three steps.  After these technical
preparations, \cref{thm:main1} is proved in \cref{sec:main1}, and
\cref{thm:main2} is proved in \cref{sec:main2}.  In the last section
(\cref{sec:isol}), we show that probabilistic arithmetic and tropical
circuits can also be derandomized using elementary arguments by using
so-called ``isolating sets'' for arithmetic and tropical polynomials.

\section{Related work}
\label{sec:previous}

As we mentioned at the beginning, our starting point is the result of
Adleman~\cite{adleman} that\footnote{Actually, the result is stronger,
  and should be stated as ``$\cPolyBPP = \cPoly$:'' even probabilistic
  \emph{circuits}, not only probabilistic Turing machines
  (\emph{uniform} sequences of circuits) can be derandomized. We,
  however, prefer to use the less precise but more familiar shortcut
  ``$\cBPP\subseteq \cPoly$.'' } $\cBPP\subseteq \cPoly$ holds for
\emph{Boolean} circuits.  In fact, Adleman proved this only when
\emph{one-sided} error is allowed. To prove the two-sided error
version, Bennett and Gill~\cite{bennett} used a simple ``finite
majority rule'' (\cref{clm:fin-maj} in \cref{sec:maj}).  This rule
follows directly from the Chernoff and union bounds, and allows us to
simulate any probabilistic circuit of size $\size$ on $n$ input
variables taking their values in a \emph{finite} domain $D$ as a
majority vote of $\bigO(n\log|D|)$ deterministic circuits, each of
size at most~$\size$.

In the \emph{Boolean} case, the domain $D=\{0,1\}$ is clearly finite,
and the majority vote functions turn into Boolean majority functions:
output $1$ if and only if more than half of the input bits are
$1$s. Since majority functions have small Boolean circuits, even
monotone ones, the resulting deterministic circuits are then not much
larger than the probabilistic ones, is only $\bigO(n\size)$.

Using entirely different arguments (not relying on the finite majority
rule), Ajtai and Ben-Or~\cite{ajtaiB84} have shown that
$\cBPP\subseteq \cPoly$ holds also for Boolean constant-depth
circuits, known also as AC$^0$ circuits. Note that this extension is
far from being trivial, because the majority function itself requires
AC$^0$ circuits of exponential size.

Markov~\cite{markov} has found a surprisingly tight combinatorial
characterization of the minimum number of NOT gates required by
\emph{deterministic} $(\lor,\land,\neg)$ circuits to compute a given
Boolean functions $f$.  A natural question therefore was: can
randomness substantially reduce the number of NOT gates?
Morizumi~\cite{morizumi} has shown that Markov's result itself already
gives a negative answer: in probabilistic circuits, the decrease of
the number of NOT gates is at most by an \emph{additive} constant,
where the constant depends only on the success probability.

The derandomization of circuits working over \emph{infinite} domains
$D$, such as $\NN$, $\ZZ$ or $\RR$, is a more delicate task. Here we
have to somehow ``cope'' with the infinity of the domain: Chernoff's
and union bounds alone do not help then. Two general approaches
emerged along this line of research.

\begin{itemize}
\item[(A)] Find (or just prove a mere existence of) a \emph{finite}
  set $X\subset D^n$ of input vectors that is ``isolating'' in the
  following sense: if a (deterministic) circuit computes a given
  function $f$ correctly on all inputs $x\in X$, then it must compute
  $f$ correctly on \emph{all} inputs $x\in D^n$. Then use the finite
  majority rule on inputs from~$X$.

\item[(B)] Use the ``infinite majority rule'' (\cref{infmajrule}
  below) following from the uniform convergence in probability
  results, proved by researchers in the statistical learning theory.
\end{itemize}

Approach (A) was used by many authors to show the inclusion
$\cBPP\subseteq \cPoly$ for various types of decision trees.  The
complexity measure here is the depth of a tree. These trees work over
$\RR$, and branch according to the sign of values of rational
functions. In the case when only linear functions are allowed for
branching, the inclusion $\cBPP\subseteq \cPoly$ was proved by Manber
and Tompa~\cite{manber}, and Snir~\cite{snir}. Meyer auf der
Heide~\cite{meyer} proved the inclusion $\cBPP\subseteq \cPoly$ for
the decision tree depth when arbitrary rational functions are
allowed. He uses a result of Milnor~\cite{milnor} about the number of
connected components of polynomial systems in $\RR^n$ to upper-bound
the minimum size of an ``isolating'' subset $X\subset\RR^n$. Further
explicit lower bounds on the depth of probabilistic decision trees
were proved by B\"urgisser, Karpinski and Lickteig~\cite{buergisser},
Grigoriev and Karpinski~\cite{GKa}, Grigoriev~et.~al.~\cite{GKMhS},
Grigoriev~\cite{Grigoriev99} and other authors.

Approach (B) was used by Cucker~et.~al.~\cite{cucker} to prove the
inclusion $\cBPP\subseteq \cPoly$ for circuits over the basis
$(+,-,\times,\div,\ssgn)$, that is, for arithmetic $(+,-,\times,\div)$
circuits with signum gates.  They combined the upper bound on the
Vapnik--Chervonenkis dimension (VC dimension) of such circuits,
obtained by Goldberg and Jerrum~\cite{goldberg}, with a uniform
convergence in probability theorem of Haussler~\cite{haussler} for
classes of functions with bounded VC dimension. In the proofs of
\cref{thm:main1,thm:main2} we will also use Approach~(B), but in a
more \emph{direct} way avoiding the detour through VC dimension and
Sauer's lemma: we will directly combine the classical uniform
convergence in probability theorem of Vapnik and
Chervonenkis~\cite{vapnik} with the upper bound of
Warren~\cite{warren} on the number of sign patterns of real
polynomials.

The $\cBPP$ vs. $\cP$ problem in the \emph{uniform} setting, that is,
in terms of Turing machines (instead of circuits), is an even more
delicate task.  Still, a strong indication that $\cBPP=\cP$ ``should''
hold also in the uniform setting was given by Impagliazzo and
Wigderson~\cite{RusselAvi}: either $\cBPP=\cP$ holds or \emph{every}
decision problem solvable by deterministic Turing machines in time
$2^{\bigO(n)}$ can be solved by a Boolean circuit of sub-exponential
size~$2^{o(n)}$.  Goldreich~\cite{Goldreich11} related the $\cBPP$
vs. $\cP$ problem with the existence of pseudorandom generators:
$\cBPP = \cP$ if and only if there exists suitable pseudorandom
generators; the ``if'' direction was known for decades---the novelty
is in the converse direction.

\section{Preliminaries}
\label{sec:semialg}

In this section, we define more precisely the concepts used in the
paper (probabilistic circuits, algebraic formulas, description
complexity of sets and functions), and recall two known results used
in the proofs of our main results (\cref{thm:main1,thm:main2}).

\subsection{Probabilistic circuits}

A circuit \emph{basis} is any family $\basis$ of multivariate
real-valued functions.  A~\emph{circuit} over a basis $\basis$ is a
sequence $\F=(f_1,\ldots,f_{\size})$ of real-valued functions, where
each $f_i$ is obtained by applying one of the basis operations to the
functions in $\RR\cup\{x_1,\ldots,x_n,f_1,\ldots,f_{i-1}\}$; scalars
$a\in \RR$ can be also viewed as (constant) functions.  The
\emph{size} of a circuit is the number $\size$ of functions in the
sequence, and the function $f:\RR^n\to \RR$ \emph{computed} by the
circuit is the last function $f=f_{\size}$ in the sequence.  Every
circuit can be also viewed as a directed acyclic graph; parallel edges
joining the same pair of nodes are allowed. Each indegree-zero node
holds either one of the variables $x_1,\ldots,x_n$ or a scalar
$a\in\RR$.  Every other node, a \emph{gate}, performs one of the
operations $g\in\basis$ on the results computed at its input gates.  A
circuit is $\dgates$-\emph{semialgebraic} if each its basis operation
(a gate) is $\dgates$-semialgebraic.

A \emph{probabilistic circuit} is a deterministic circuit which,
besides the actual (deterministic) variables $x_1,\ldots,x_n$, is
allowed to use additional variables
$\xr_1,\ldots,\xr_k$, each being a \emph{random} variable taking its
values in $\RR$. As we already mentioned in the introduction, the
probability distribution of these random variables can be arbitrary:
our derandomization results will hold for \emph{any} distribution.

\subsection{Semialgebraic sets and functions}
A set $S\subseteq \RR^n$ is \emph{semialgebraic} if it can be obtained
by finitely many unions and intersections of sets defined by a
polynomial equality or strict inequality. For us important will be not
the mere fact that a set $S$ \emph{is} semialgebraic but rather ``how
much semialgebraic'' it actually is: how many distinct polynomials and
of what degree do we need to define this set?

The \emph{sign function} $\ssgn:\RR\to\{-1,0,+1\}$ takes value
$\sgn{x}=-1$ if $x< 0$, $\sgn{0}=0$, and $\sgn{x}=+1$ if $x>0$. Let
$P=(p_1,\ldots,p_m)$ be a sequence of polynomials in
$\RR[x_1,\ldots,x_n]$.  The \emph{sign pattern} of this sequence at a
point $x\in\RR^n$ is the vector
\[
\sgn{P(x)} = \big(\sgn{p_1(x)},\ldots,\sgn{p_m(x)} \big)
\in\{-1,0,+1\}^n
\]
of signs taken by these polynomials at the point~$x$.

A set $S\subseteq\RR^n$ is $\slocal$-\emph{semialgebraic} if there is
a sequence $P=(p_1,\ldots,p_m)$ of $m\leq \slocal$ polynomials of
degree at most $\slocal$ such that the membership of points $x\in
\RR^n$ in the set $S$ can be determined from sign patterns of these
polynomials at these points, that is, if $x\in S$ and $x'\not\in S$,
then $\sgn{P(x)}\neq\sgn{P(x')}$.

A function $f:\RR^n\to\RR^m$ is $\slocal$-\emph{semialgebraic} if its
graph $S=\{(x,y)\colon y=f(x)\}\subseteq\RR^{n+m}$ is such. The
\emph{description complexity} of a semialgebraic set (or function) is
the smallest number $\slocal$ for which this set (or function) is
$\slocal$-semialgebraic.

\subsection{Algebraic formulas}
The description complexity of sets and functions can be defined more
explicitly using the language of ``algebraic formulas.''  An
\emph{algebraic formula} is an arbitrary Boolean combination of atomic
predicates, each being of the form $\pred{p(x)\ \Diam\ 0}$ for some
polynomial $p$ in $\RR[x_1,\ldots,x_n]$, where $\Diam$ is one of the
standard relations $>$, $\geq$, $=$, $\neq$, $\leq$, $<$; the
predicate $\pred{\rho}$ for a relation $\rho$ outputs $1$ if the
relation $\rho$ holds, and outputs $0$ otherwise.  So, for example,
$\pred{p(x)=0}=1$ if and only if $p(x)=0$.  Note that $\pred{p(x)\
  \Diam\ q(x)}$ is equivalent to $\pred{p(x)-q(x)\ \Diam\ 0}$, so that
we can also make comparisons between polynomials.  The
\emph{description complexity} of an algebraic formula is
$\max\{m,d\}$, where $m$ is the number of distinct polynomials used in
the formula, and $d$ is their maximal degree. An algebraic formula
$\Phi(x)$ \emph{defines} a set $S\subseteq \RR^n$ if
$S=\{x\in\RR^n\colon \Phi(x)=1\}$.

\begin{clm}\label{clm:normal}
  For every algebraic formula there is a algebraic formula of the same
  description complexity which only uses atomic predicates of the form
  $\pred{p< 0}$, $\pred{p=0}$ and $\pred{p>0}$.
\end{clm}
The claim is trivial: just replace each atomic predicate $\pred{p\leq
  0}$ by the formula $\pred{p=0}\lor \pred{p<0}$, each atomic
predicate $\pred{p\geq 0}$ by the formula $\pred{p=0}\lor \pred{p>0}$,
and each atomic predicate $\pred{p\neq 0}$ by the formula
$\pred{p<0}\lor \pred{p>0}$. Neither the number of distinct
polynomials used, nor their degree increases during these
transformations.

\begin{clm}\label{clm:equivalence}
  The description complexity of a semialgebraic set is the minimum
  description complexity of an algebraic formula defining this set.
\end{clm}

In the literature, this fact is often used as the \emph{definition} of
the description complexity of sets.

\begin{proof}
  Let $S\subseteq \RR^n$ be a set of vectors. Our goal is to show that
  the description complexity of $S$ is at most $\slocal$ if and only
  if the set $S$ can be defined by an algebraic formula $\Phi$ of
  description complexity at most~$\slocal$.

  $(\Leftarrow)$ By \cref{clm:normal}, we can assume that only atomic
  predicates of the form $\pred{p< 0}$, $\pred{p=0}$ and $\pred{p>0}$
  are used in the formula $\Phi$.  Hence, the values of the formula
  $\Phi$ only depend on the sign patterns of the sequence
  $P=(p_1,\ldots,p_m)$ of all $m\leq \slocal$ polynomials of degree at
  most $\slocal$ used in the formula $\Phi$.

  $(\Rightarrow)$ Let $P=(p_1,\ldots,p_m)$ be a sequence of $m\leq
  \slocal$ polynomials of degree at most $\slocal$ such that the
  membership of points $x\in \RR^n$ in the set $S$ can be determined
  from sign patterns of these polynomials on these points, that is,
  $\sgn{P(x)}\neq\sgn{P(x')}$ holds for every two points $x\in S$ and
  $x'\not\in S$.  Given a sign pattern $\sigma\in\{-1,0,+1\}^m$, let
  \[
  M_{\sigma}(x)=\bigwedge_{i=1}^m\pred{p_i(x)\ \Diam_i\ 0}
  \]
  be the AND of atomic predicates corresponding to $\sigma$, where
  $\Diam_i$ is ``$<$'' if $\sigma_i=-1$, is ``$=$'' if $\sigma_i=0$,
  and is ``$>$'' if $\sigma_i=+1$. Note that, for every point
  $x\in\RR^n$, we have $M_{\sigma}(x)=1$ if $\sgn{P}(x)=\sigma$, and
  $M_{\sigma}(x)=0$ if $\sgn{P}(x)\neq \sigma$.  So, if
  $\Sigma=\{\sgn{P}(x)\colon x\in S\}$ is the set of all sign patterns
  of the sequence $P$ of polynomials on the points in the set $S$,
  then the algebraic formula
  \[
  \Phi(x)=\bigvee_{\sigma\in \Sigma} M_{\sigma}(x)
  \]
  defines the set~$S$.
\end{proof}

By \cref{clm:equivalence}, a function $f:\RR^n\to\RR$ is
$\slocal$-semialgebraic if there is an algebraic formula $\Phi(x,y)$
of description complexity at most $\slocal$ such that for every
$x\in\RR^n$ and $y\in\RR$, $\Phi(x,y)=1$ holds precisely when
$y=f(x)$.  \Cref{tab:basic} gives a sample of some basic semialgebraic
functions of small description complexity.

\tabelle

Let us stress that, in algebraic formulas, we only count the number of
\emph{distinct} polynomials used, \emph{not} the number of their
\emph{occurrences} in the formula: one and the same polynomial can
appear many times, and under different relations $\Diam$.

\begin{ex}[Sorting operation]
  The \emph{sorting operation} $\sort:\RR^n\to\RR^n$ takes a sequence
  $x_1,\ldots,x_n$ of real numbers, and outputs its ordered
  permutation $y_1\leq y_2\leq \ldots\leq y_n$. The graph of this
  operation can be defined by the following algebraic formula of $2n$
  variables:
  \[
  \Phi(x,y)= \bigwedge_{i=1}^{n-1}\pred{y_i\leq
    y_{i+1}}\land\Big(\bigvee_{\sigma\in
    S_n}\bigwedge_{i=1}^n\pred{y_i=x_{\sigma(i)}}\Big)\,,
  \]
  where $S_n$ is the set of all permutations of $\{1,\ldots,n\}$.  The
  total number of occurrences of atomic predicates in this formula
  (the ``size'' of the formula) is huge (is even larger than $n!$),
  but the formula only uses $m=n^2+n-1$ distinct polynomials
  $y_{i+1}-y_i$ for $i=1,\ldots,n-1$, and $y_i-x_j$ for
  $i,j=1,\ldots,n$ of degree~$d=1$. Thus, the sorting operation
  $\sort:\RR^n\to\RR^n$ is $\slocal$-semialgebraic for
  $\slocal=n^2+n-1$.
\end{ex}

\subsection{Quantifier elimination}

A \emph{quantified algebraic formula} $\Psi(x)$ with $n$ free
variables $x=(x_1,\ldots,x_n)$ is of a form
\begin{equation}\label{eq:quant}
  (Q_1\ z_1\in\RR^{k_1})\ \ldots\ (Q_{\omega}\ z_{\omega}\in\RR^{k_{\omega}}) \
  \Phi(x,z_1,\ldots,z_{\omega})\,,
\end{equation}
where $Q_i\in\{\exists,\forall\}$, $Q_i\neq Q_{i+1}$, and $\Psi$ is an
algebraic formula using $m$ polynomials of degree at most~$d$. That
is, we have $\omega$ alternating blocks of quantifiers with $k_i$
quantified variables in the $i$th block.

\begin{thm}[Basu, Pollack and Roy~\cite{basu}]
  \label{thm:BPR}
  For every quantified formula \cref{eq:quant} there is an equivalent
  quantifier-free algebraic formula $\Phi(x_1,\ldots,x_n)$ which uses
  at most $(sd)^{\bigO(nK)}$ polynomials of degree at most
  $d^{\bigO(K)}$, where $K=\prod_{i=1}^{\omega}k_i$.
\end{thm}

We will only use a very special consequence of this result for
\emph{existential} algebraic formulas, that is, for formulas
\cref{eq:quant} with $\omega=1$, $Q_1=\exists$ and $k_1=q$.

\begin{cor}\label{cor:elimination}
  If $S=\{x\in\RR^n\colon (\exists z\in\RR^q)\ \Phi(x,z)=1\}$ for some
  quantifier-free algebraic formula $\Phi$ of description complexity
  $\kappa$, then $S$ is $\slocal$-semialgebraic for $\log
  \slocal=\bigO(nq\log\kappa)$.
\end{cor}

\subsection{Number of sign patterns}
\label{sec:zerro}
By the definition, a set $S\subseteq\RR^n$ is $\slocal$-semialgebraic
if the membership of points $x\in\RR^n$ in $S$ can be determined from
seeing the sign patterns of some fixed sequence of $\slocal$
polynomials of degree at most $\slocal$ on these points $x$. So, a
natural question arises: how many distinct sign patterns a given
sequence of $m$ polynomials on $n$ variables can have?  A trivial
upper bound is $|\{-1,0,+1\}^m|=3^m$.

A fundamental result of Warren~\cite[Theorem~3]{warren} shows that,
when we have more than $n$ polynomials of bounded degree, then the
critical parameter is not their number $m$ but rather the number $n$
of variables.

\begin{thm}[Warren~\cite{warren}]\label{thm:warren}
  No sequence of $m\geq n$ polynomials in $\RR[x_1,\ldots,x_n]$ of
  degree at most $d\geq 1$ can have more than $(8\euler md/n)^n$
  distinct sign patterns.
\end{thm}

What Warren actually proved is the upper bound $(4\euler md/n)^n$ on
the number of sign patterns lying in the set $\{-1,+1\}^n$. But as
observed in \cite{alon-scheinerman,PudlakR92,goldberg}, by
``doubling'' each polynomial, this bound can be easily extended to the
upper bound $(8\euler md/n)^n$ on the number of \emph{all} sign
patterns. To see this, let $p_1,\ldots,p_m$ be a sequence of
polynomials in $\RR[x_1,\ldots,x_n]$ of degree at most $d$. The
sequence can clearly have at most $3^m$ distinct sign patterns. So,
there is a \emph{finite} set $X\subset\RR^n$ of $|X|\leq 3^m$ vectors
witnessing all distinct sign patterns of this sequence. Take
\[
\err=\tfrac{1}{2}\cdot\min\{p_i(x)\colon \mbox{$x\in X$ and
  $p_i(x)\neq 0$}\}\,,
\]
and consider the sequence $p_1-\err,p_1+\err,\ldots,p_m-\err,p_m+\err$
of $2m$ polynomials. By the choice of $\err$, each two distinct
$(-1,0,+1)$ patterns of the original sequence lead to also distinct
$(-1,+1)$ patterns of the new sequence.

We will use the following direct consequence of Warren's theorem.

\begin{cor}\label{cor:warren}
  Let $\Phi_1(x),\ldots,\Phi_m(x)$ be a sequence algebraic formulas on
  the same $n$ variables. If each of these formulas have description
  complexity at most $\slocal$, then
  \[
  \Big|\big(\Phi_1(x),\ldots,\Phi_m(x)\big)\colon x\in\RR^n\Big|\leq
  \left(\frac{8\euler m \slocal^2}{n}\right)^n\,.
  \]
\end{cor}
This follows from \cref{thm:warren} because the values of any such
sequence of algebraic formulas only depend on the sign patterns of the
sequence of $\leq m\slocal$ polynomials of degree $\leq \slocal$ used
in these formulas.

\subsection{What functions are not semialgebraic?}

To show what kind of operations we do \emph{not} allow to be used as
gates, let us recall the following well known \emph{necessary}
condition for a set to be semialgebraic.

\begin{clm}\label{clm:necess}
  If a set $S\subseteq\RR^n$ is semialgebraic, then either the
  interior of $S$ is nonempty, or some nonzero polynomial vanishes on
  all points of~$S$.
\end{clm}

\begin{proof}
  By observing that a system of equations $p_1(x)=0,\ldots, p_m(x)=0$
  is equivalent to one equation $p_1(x)^2+\cdots+p_m(x)^2=0$, and that
  $p(x)<0$ is the same as $-p(x)>0$, we have that a set
  $S\subseteq\RR^n$ is semialgebraic if and only if it is a finite
  union $S=S_1\cup S_2\cup \cdots\cup S_m$ of (nonempty) sets of the
  form $S_i=\{x\in\RR^n\colon p_i(x)=0,
  q_{i,1}(x)>0,\ldots,q_{i,k_i}(x)>0\}$, where $p_i$ and $q_{i,j}$ are
  real polynomials. So, if some $p_i$ is the zero polynomial, then $S$
  has a nonempty interior. Otherwise, $p_1\cdot p_2\cdots p_m$ is a
  nonzero polynomial vanishing on all points of~$S$.
\end{proof}

\begin{ex}\label{ex:rounding}
  \Cref{clm:necess} can be used to show that some functions are
  \emph{not} semialgebraic.  Consider, for example, the rounding
  function $f(x)=\lfloor x\rfloor$. That is, for a real number
  $x\in\RR$, $f(x)$ is the largest integer $n$ such that $n\leq x$.
  The interior of the graph $S=\{(x,y)\in \RR\times \ZZ\colon \lfloor
  x\rfloor=y\}$ of $\lfloor x\rfloor$ is clearly empty, because $y$
  can only take integer values. But the only polynomial
  $p(x,y)=\sum_{i=0}^d p_i(y)\cdot x^i$ vanishing on all points of $S$
  must be the zero polynomial. Indeed, since $p$ vanishes on $S$, the
  polynomial $p(x,n)$ has an infinite (and, hence, larger than $d$)
  number of roots $x\in[n,n+1)$, for every integer $n$; so, $p_i(n)=0$
  for all $i$. Since this holds for infinitely many numbers $n$, all
  polynomials $p_0,p_1,\ldots,p_d$ must be zero polynomials. So, the
  rounding function is not semialgebraic.
\end{ex}

\section{The route to derandomization}
\label{sec:enroute}

In our derandomization of probabilistic circuits, the following
parameters of (finite or infinite) Boolean matrices $M:A\times
B\to\{0,1\}$ will be crucial.
\begin{itemize}
\item The matrix $M$ has the $m$-\emph{majority property} if there is
  a sequence $b_1,\ldots,b_m\in B$ of not necessarily distinct columns
  of $M$ such that
  \[
  M[a,b_1]+\cdots+ M[a,b_m] > m/2
  \]
  holds for every row $a\in A$.

\item The matrix $M$ is \emph{probabilistically dense} if there exists
  a probability distribution $\pprob:B\to[0,1]$ on the set of columns
  such that
  \[
  \prob{b\in B\colon M[a,b]=1}\geq 2/3
  \]
  holds for every row $a\in A$.  Note that the mere \emph{existence}
  of at least one probability distribution with this property is
  sufficient. Thus, density is a property of matrices, not of
  probability distributions on their columns.

\item The \emph{growth function} of $M$ is the function
  $\Pi_M:\NN\to\NN$ whose value $\Pi_{M}(m)$ for each integer $m\geq
  1$ is the maximum
  \[
  \Pi_M(m)=\max_{b_1,\ldots,b_m}\big|\left\{\big(M[a,b_1]\ldots,M[a,b_m]\big)\colon
    a\in A\right\} \big|
  \]
  over all choices of $m$ columns of $M$, of the number of distinct
  $0$-$1$ patterns from $\{0,1\}^m$ appearing as rows of $M$ in these
  columns. Note that $1\leq \Pi_M(m)\leq 2^m$ for every $m\geq 1$.
  Let us mention that the maximum number $m$ (if there is one) for which $\Pi(m)=2^m$
  holds is know as the \emph{Vapnik--Chervonenkis dimension} of the matrix~$M$.

\item A matrix $M:\RR^n\times \RR^k\to\{0,1\}$ is \emph{semialgebraic}
  if the set $ S=\{(x,y)\in\RR^{n+k}\colon M[x,y]=1\} $ of its
  $1$-entries is such. The \emph{description complexity} of a column
  $r\in\RR^k$ is the description complexity of the set
  $S_r=\{x\in\RR^n\colon M[x,r]=1\}$ of its $1$-entries.

\end{itemize}
Given a probabilistic circuit $\F(x,\xr)$ computing a given function
$f:\RR^n\to\RR$, the following two Boolean matrices naturally arise,
where $k$ is the number of random input variables
$\xr=(\xr_1,\ldots,\xr_k)$.

\begin{itemize}
\item The \emph{graph matrix} of $\F(x,\xr)$ is the Boolean matrix
  $M_F\colon \RR^{n+1}\times \RR^k\to\{0,1\}$ with entries defined by:
  \[
  \mbox{$M_F[(x,y),\rr]=1$ if and only if $\F(x,\rr)=y$.}
  \]
  The graph matrix $M_F$ gives us a full information about all
  functions computed by the deterministic circuits $\F(x,\rr)$ obtained from
  $\F(x,\xr)$ by setting the random inputs $\xr$ of $\F$ to all
  possible values $\rr\in\RR^k$.

\item The \emph{correctness matrix} of $\F(x,\xr)$ with respect to the
  given function $f:\RR^n\to\RR$ is the Boolean matrix $M:\RR^n\times
  \RR^k\to\{0,1\}$ with entries defined by:
  \[
  \mbox{$M[x,\rr]=1$ if and only if $\F(x,\rr)=f(x)$.}
  \]
  Note that $M$ is a \emph{submatrix} of the graph matrix $M_F$: just
  remove all rows of $M_F$ labeled by pairs $(x,y)$ such that $y\neq
  f(x)$, and replace the label $(x,y)$ of each remaining row by~$x$.
\end{itemize}
The relation of the majority property of matrices to the
derandomization of probabilistic circuits is quite natural. Suppose
that a probabilistic circuit $\F(x,\xr)$ computes the correct values
$f(x)$ of a given function $f$ with probability $\geq 2/3$.  Then the
correctness matrix $M$ is probabilistically dense \emph{per se}.  On
the other hand, if the matrix $M$ has the $m$-majority property, then
there are $m$ (not necessarily distinct) assignments
$\rr_1,\ldots,\rr_m\in\RR^k$ to the random input variables such that,
for every input $x\in\RR^n$, the \emph{deterministic} circuit
$F(x)=\maj(\F(x,\rr_1),\ldots,\F(x,\rr_m))$ outputs the correct
value~$f(x)$.

Thus, the derandomization of probabilistic circuits boils down to
showing that their correctness matrices have the $m$-majority property
for possibly small values of~$m$. We will show this in the following
three steps, where $\F(x,\xr)$ is a probabilistic circuit with $n$
deterministic input variables, $k$ random input variables, and $\size$
gates.

Let $\slocal$ be the minimal number $t$ such that for every
$r\in\RR^k$, the function $\F_r:\RR^n\to\RR$ computed by the
deterministic circuit $\F_r(x)=\F(x,r)$ is $\slocal$-semialgebraic.

\begin{description}
\item[Step~1] (\cref{lem:growth} in \cref{sec:growth}) The growth
  function $\Pi_{M_F}(m)$ of the graph matrix $M_F$ of $\F$ satisfies
  \[
  \ln \Pi_{M_F}(m)\leq n\ln (8\euler m \slocal^2/n)=2n\ln\slocal +
  n\ln (8\euler m/n)\,.
  \]

\item[Step~2] (\cref{lem:VC} in \cref{sec:maj}) There is an absolute
  constant $\const>0$ such that every probabilistically dense
  submatrix of $M_F$ has the $m$-majority property for any $m\geq
  2/\const$ satisfying
  \[
  \ln \Pi_{M_F}(m)\leq \const m\,.
  \]

\item[Step~3] (\cref{lem:descrcomplexity} in \cref{sec:quantifiers})
  If the description complexity of each single gate of $\F$ does not
  exceed $\dgates$, then
  \[
  \ln \slocal =\bigO(n\size\ln \dgates\size)\,.
  \]
\end{description}
Now, if the probabilistic circuit $\F(x,\xr)$ computes a given
function $f:\RR^n\to\RR$, then the correctness matrix $M$ of this
circuit with respect to the function $f$, \emph{is} a
probabilistically dense \emph{per se}. Also, as we have shown right
after its definition, the correctness matrix $M$ is a \emph{submatrix}
of the graph matrix $M_F$ of the circuit $F$.  Thus, by Steps~1--3,
the matrix $M$ has the $m$-majority property for any $m$ satisfying
the inequality $2n\ln\slocal + n\ln (8\euler m/n)\leq \const m$, where
$2n\ln\slocal=\bigO(n^2\size\log\dgates\size)$. This inequality is
satisfied by taking $m=Cn^2\size\log\dgates\size$ for a sufficiently
large (but absolute) constant~$C$.

The case of \emph{approximating} (not necessarily exactly computing)
probabilistic circuits requires an additional idea. The reason is that
then the ``approximate correctness'' matrix $M$ of the circuit $\F$
approximating the function $f$ is \emph{not} necessarily a submatrix
of the graph matrix $M_F$ of the circuit $F$. For example, if
$\F(x,\rr)=z$ for some $z$ such that $z\neq f(x)$ but $z\rel f(x)$,
then $M_F[(x,f(x)),\rr]=0$ but the corresponding entry $(x,\rr)$ in
the ``approximate correctness'' matrix $M$ will then be
$M[x,\rr]=1$. This is why in \cref{thm:main2}, unlike in
\cref{thm:main1}, also the description complexities $\relcompl$ and
$\sslocal{f}$ of the approximation relation $\rrel$ and of the
approximated function $f$ come to play.

We now turn to detailed proofs.

\section{Step 1: Growth functions from description complexity}
\label{sec:growth}

\begin{lem}
  \label{lem:growth}
  Let $M:\RR^n\times \RR^k\to\{0,1\}$ be a Boolean matrix. If the
  description complexity of every column of $M$ does not exceed
  $\slocal$, then for all $m\geq n$, the growth function $\Pi_M(m)$ of
  $M$ satisfies
  \[
  \Pi_M(m)\leq \left(\frac{8\euler m \slocal^2}{n}\right)^n\,.
  \]
\end{lem}

\begin{proof}
  Take arbitrary $m$ columns $r_1,\ldots,r_m\in \RR^k$ of $M$.  Since
  every column of $M$ is $\slocal$-semialge\-bra\-ic, for every
  $i=1,\ldots,m$ there is an algebraic formula $\Phi_i(x)$ which uses
  at most $\slocal$ distinct polynomials of degree at most $\slocal$,
  and satisfies $M[x,r_i]=\Phi_i(x)$ for all $x\in \RR^n$.  So,
  $\Pi_M(m)$ is at most the number of distinct $0$-$1$ strings
  $(\Phi_1(x),\ldots,\Phi_m(x)\big)$ when $x$ ranges over the entire
  set $\RR^n$ of row labels of the matrix $M$. By \cref{cor:warren},
  the number of such strings is at most $\left(8\euler m
    \slocal^2/n\right)^n$.
\end{proof}

\begin{rem}\label{rem:columns}
  In the case when the description complexity of the \emph{entire}
  matrix $M$ is bounded by $\slocal$, similar upper bounds on the
  growth function were already derived by
  Goldberg and Jerrum~\cite{goldberg} from Warren's theorem, and by Ben-David and
  Lindenbaum~\cite{lindenbaum} from a similar result of Milnor~\cite{milnor}. Our observation is that the same upper
  bound actually holds when only the description complexities of
  \emph{individual} columns are bounded by $\slocal$. In the context
  of derandomization, this will alow us to make the blowup in size of
  derandomized circuits \emph{independent} on the number $k$ of random
  input variables (note that $k$ may be as large as the size of the
  probabilistic circuits).

  This observation also extends the bound to a properly larger class
  of matrices.  The point is that the description complexity of
  individual columns may be much smaller than that of the entire
  matrix $M$. Even more: the former can even then be small, when the
  entire matrix is \emph{not} semialgebraic at all (its description
  complexity is unbounded).  As a trivial example, consider the matrix
  $M:\RR\times\RR\to\{0,1\}$ whose entries are defined by: $M[x,r]=1$
  if and only if $x=\lfloor r\rfloor$. The matrix is not semialgebraic
  (see \cref{ex:rounding}), but for every \emph{fixed} column
  $\rr\in\RR$, the set of $1$-entries of the $\rr$th column is defined
  by a semialgebraic formula $\pred{x-c=0}$, where $c=\lfloor
  \rr\rfloor$ is a (fixed) integer. Hence, the description complexity
  of each individual column is~$1$.
\end{rem}

\section{Step 2: Majority property from growth functions}
\label{sec:maj}

As we mentioned in \cref{sec:enroute}, the derandomization of
probabilistic circuits boils down to showing that their correctness
matrices have the $m$-majority property for possibly small values
of~$m$.

\subsection{Finite majority rule}
The following ``folklore'' observation shows that, if the number of
rows is \emph{finite}, then the $m$-majority property holds already
for $m$ about the logarithm of this number.

\begin{lem}[Finite majority rule]\label{clm:fin-maj}
  Every probabilistically dense Boolean matrix $M:A\times B\to\{0,1\}$
  with a finite number $|A|$ of rows has the $m$-majority property for
  $m=O(\log|A|)$. In particular, at least one column of $M$ has more
  than $|A|/2$ ones.
\end{lem}

\begin{proof}
  Since the matrix $M$ is probabilistically dense, we know that there
  is a probability distribution $\pprob:B\to[0,1]$ such that
  $\prob{b\in B\colon M[a,b]=1}\geq 2/3$ holds for every row $a\in A$.
  Let $b_1,\ldots,b_m$ be $m$ independent copies of~$b$.  The expected
  value $\mu$ of the sum $\rv=M[a,b_1]+\cdots+M[a,b_m]$ is at least
  $2m/3$.  Thus, for $t:=m/6$, the event $\rv\leq m/2$ implies the
  event $\rv\leq \mu - t$.  By the Chernoff--Hoeffding bound (see, for
  example,~\cite[Theorem~1.1]{dubhashi}), the probability of the
  latter event is at most $\euler^{-2t^2/m}=\euler^{-2(m/6)^2/m}<
  \euler^{-m/27}$.  By taking $m=\lceil 27\log|A|\rceil$, this
  probability is strictly smaller than $1/|A|$. Since we only have
  $|A|$ rows, the union bound implies that the matrix $M$ has the
  $m$-majority property for this value of~$m$. This shows the first
  claim. The second claim follows by double-counting: the number of
  ones in the corresponding $m$ columns of $M$ is $>(m/2)|A|$. So, at
  least one of these columns must have $>|A|/2$ ones.
\end{proof}

\Cref{clm:fin-maj} allows us to derandomize probabilistic circuits
working over any \emph{finite} domain (including Boolean circuits): if
the probabilistic circuit has size $s$, then the obtained
deterministic circuit (with one additional majority vote operation as
the output gate) will have size $O(ns)$.  We are, however, interested
in circuits simulating dynamic programming algorithms. These circuits
work over \emph{infinite} (or even uncountable) domains like $\NN$,
$\ZZ$, $\QQ$ or $\RR$; elements of the domain are possible weights of
items in optimization problems. So, in this case, the finite majority
rule is of no use at all.

\subsection{Uniform convergence in probability}
Fortunately, results from the statistical learning theory come to
rescue. The classical \emph{uniform convergence in probability}
theorem of Vapnik and Chervonenkis~\cite{vapnik} ensures the majority
property also for boolean matrices $M$ with an infinite (and even uncountable) numbers of rows, as long as their growth functions $\Pi_M(m)$ grow not too fast
(\cref{lem:VC} below).

Let $H$ be a class of $0$-$1$ functions $h:X\to\{0,1\}$ on a set $X$,
and $\pprob: X\to[0,1]$ be a probability distribution on the set
$X$. Draw independently (with repetitions) a sequence
$\xx=(\xx_1,\ldots,\xx_m)$ of samples $\xx_i\in X$ according to this
probability distribution. The \emph{empirical frequency} of $h\in H$
on $\xx$ is the average value
\[
\ef{h}{\xx} :=\frac{h(\xx_1)+\cdots+ h(\xx_m)}{m}\,,
\]
while the \emph{theoretical probability} of the function $h$ itself is
its expected value
\[
\tp{h}:= \prob{x\in X\colon h(x)=1}\,.
\]
Every function $h:X\to\{0,1\}$ defines the event $A=\{x\in X\colon
h(x)=1\}$.  The law of large numbers says that, for each single event,
its empirical frequency in a sequence of independent trials converges
(with high probability) to its theoretical probability. We, however,
are now interested not in a single event but in a whole family of
events. We would like to know whether the empirical frequency of every
event in the family converges to its theoretical probability
\emph{simultaneously}.  This is the content of so-called ``uniform
convergence in probability'' results in statistics.

In these results (including \cref{thm:VC} below), a natural \emph{permissibility}
assumption is made when the class $H$ is \emph{uncountable}. The permissibility of $H$ means that the individual functions $h\in H$ as well as the supremum function $\dist{x}=\sup_{h\in H} |\ef{h}{\xx} -\tp{h}|$ are
  measurable. That is, we need that for a random sample $\xx\in X^m$,
  $\dist{\xx}$ is a random variable.

  \begin{rem}\label{rem:permis}
   In our applications, the classes
  $H$ will correspond to the rows of graph matrices of semialgebraic
  circuits.  So, each class $H$ will consist of $0$-$1$ valued
  \emph{semialgebraic} functions $h:X\to\{0,1\}$, where $X=\RR^k$ for
  some finite $k\geq 1$, and will be of the form
  $H=\{\indf(t,\cdot)\colon t\in\RR^n\}$ for a finite $n\geq 1$, where
  the indexing function $\indf:\RR^n\times X\to\{0,1\}$ (the matrix
  itself) is also semialgebraic.  Such classes $H$ are permissible;
  see \cref{app:permissible} for more details. Let us also note that one can also avoid the ``permissibility issue'' by only considering circuits working over
  all rational (instead of all real) numbers: then the classes $H$ of
  functions corresponding to the rows of graph matrices of such
  circuits are \emph{countable} and, hence, permissible.
\end{rem}

The \emph{growth function} of the family $H$ is the function
$\Pi_H:\NN\to\NN$ whose value $\Pi_{H}(m)$ for each integer $m\geq 1$
is the maximum,
\[
\Pi_H(m)=\max_{x_1,\ldots,x_m}\big|\left\{\big(h(x_1),\ldots,h(x_m)\big)\colon
  h\in H\right\} \big|
\]
over all sequences $x_1,\ldots,x_m$ of (not necessarily distinct)
points in $X$, of the number of distinct $0$-$1$ patterns from
$\{0,1\}^m$ produced by the functions $h\in H$ on these points. Note
that we always have $1\leq \Pi_H(m)\leq 2^m$.

The \emph{uniform convergence theorem} of Vapnik and
Chervonenkis~\cite{vapnik} states that if the class $H$ is ``simple''
in that $\Pi_H(m)$ grows not too fast, and if we draw samples
independently (with replacement) from $X$ according to any
distribution, then with high probability, the empirical frequency
$\ef{h}{\xx} $ of \emph{every} function $h\in H$ will be close to the
theoretical probability $\tp{h}$ of~$h$.

Note that even if our family $H$ of functions $h:X\to\{0,1\}$ is
\emph{infinite} or even uncountable, there is only a \emph{finite}
number $\Pi_{H}(m)$ of their classes such that the functions lying
within the same class take the \emph{same} values on \emph{all} $m$
sampled points $\xx_1,\ldots,\xx_m\in X$. By combining this simple
observation with insightful ideas, Vapnik and
Chervonenkis~\cite{vapnik} proved the following result.

\begin{thm}[Vapnik and Chervonenkis~\cite{vapnik}]\label{thm:VC}
  Let $H$ be a permissible class of $0$-$1$ functions $h:X\to\{0,1\}$
  on a set $X$, and $\pprob: X\to[0,1]$ a probability distribution on
  the set $X$. Let $\err >0$, and draw independently (with
  repetitions) a sequence $\xx=(\xx_1,\ldots,\xx_m)$ of $m\geq
  2/\err^2$ samples $\xx_i\in X$ according to this probability
  distribution. Then
  \begin{equation}\label{eq:vc0}
    \prob{\exists h\in H\colon \left|\ef{h}{\xx}
        -\tp{h}\right|>\err} \leq 4\cdot\Pi_{H}(2m)\cdot
    \euler^{-\epsilon^2 m/8}\,.
  \end{equation}
\end{thm}

In particular, for every constant $0<\err\leq 1$ there is a constant
$\const>0$ with the following property: if the sample size $m$
satisfies
\begin{equation}\label{eq:vc1}
  m\geq 2/\const\ \ \mbox{ and }\ \  \Pi_H(m) \leq \euler^{\const m}\,,
\end{equation}
then there exists a sequence $x=(x_1,\ldots,x_m)$ of (not necessarily
distinct) points in $X$ such that $\ef{h}{x}\geq \tp{h} -\err$, that
is, the inequality
\begin{equation}\label{eq:vc2}
  h(x_1)+\cdots+h(x_m)\geq (\tp{h}-\err)m
\end{equation}
holds for \emph{all} functions $h\in H$.

We now turn back to the language of matrices.  Let $M:T\times
X\to\{0,1\}$ be a Boolean matrix.  Each row $t\in T$ of $M$ gives us a
$0$-$1$ valued function $h_t:X\to\{0,1\}$ whose values are
$h_t(x)=M[t,x]$. We say that the matrix $M$ is \emph{permissible} if
the class $H=\{h_t\colon t\in T\}$ of functions corresponding to its
rows is permissible.

Recall that the \emph{growth function} $\Pi_M(m)$ of the matrix $M$ is
the maximum, over all choices of up to $m$ columns, of the number of
distinct $0$-$1$ patterns from $\{0,1\}^m$ appearing as rows in these
columns.  Note that $\Pi_M(m)$ coincides with the growth function
$\Pi_H(m)$ of the class of functions~$H$ defined by the rows of~$M$.
In what follows, under a \emph{submatrix} of a matrix $M$ we will
understand a submatrix obtained by removing some rows of $M$; that is,
we do not remove columns.

\begin{lem}
  \label{lem:VC}
  There is an absolute constant $\const>0$ for which the following
  holds.  If a Boolean matrix $M$ is permissible, then every
  probabilistically dense submatrix of $M$ has the $m$-majority
  property for any integer $m\geq 2/\const$ satisfying $\Pi_M(m) \leq
  \euler^{\const m}$.
\end{lem}

\begin{proof}
  Let $M:T\times X\to\{0,1\}$ be a permissible matrix, and let
  $H=\{h_t\colon t\in T\}$ be the class of functions $h_t(x)=M[t,x]$
  defined by the rows $t\in T$ of $M$. Let $M'$ be any
  probabilistically dense submatrix of $M$, and $H'\subseteq H$ be the
  class of functions corresponding to the rows of $M'$. Hence, there
  is a probability distribution $\pprob:X\to[0,1]$ on the set $X$ of
  columns such that the probability $\tp{h}=\prob{x\in X\colon
    h(x)=1}$ is at least $2/3$ for every row $h\in H'$ of the
  submatrix $M'$.

  Fix $\err:=1/7$, and let $\const>0$ be a constant for which
  \cref{eq:vc1} (and, hence, also \cref{eq:vc2}) holds with this
  choice of $\err$. By \cref{eq:vc2}, there exists a sequence
  $x_1,\ldots,x_m$ of (not necessarily distinct) columns of $M$ such
  that
  \[
  h(x_1)+\cdots+h(x_m)\geq \left(\tp{h}-\err \right) m
  =\left(\tp{h}-\tfrac{1}{7}\right) m
  \]
  holds for every row $h\in H$ of $M$. For some rows $h\in H$ of $M$
  (those with $\tp{h}\leq \err$), this lower bound is trivial.  But
  since the submatrix $M'$ is probabilistically dense, we know that
  $\tp{h}\geq 2/3$ holds for all rows $h\in H'$ of this submatrix
  $M'$. Thus, for every row $h\in H'$, we have
  \[
  h(x_1)+\cdots+h(x_m)\geq \left(\tp{h}-\tfrac{1}{7} \right) m\geq
  \left(\tfrac{2}{3}-\tfrac{1}{7} \right) m = \tfrac{11}{21}m >
  \tfrac{1}{2}m\,,
  \]
  meaning that the submatrix $M'$ has the $m$-majority property, as
  desired.
\end{proof}

\subsection{Infinite majority rule}
\label{sec:inf-maj}

Recall that every probabilistically dense Boolean matrix $M:A\times
B\to\{0,1\}$ with a \emph{finite} number of rows has the $m$-majority
property for $m=\bigO(\log|A|)$. When combined with Warren's theorem
(\cref{thm:warren}), the theorem of Vapnik and Chervonenkis
(\cref{thm:VC}) yields the following result for \emph{infinite}
matrices.

\begin{lem}[Infinite majority rule]
  \label{infmajrule}
  Let $M:\RR^n\times \RR^k\to\{0,1\}$ be a semialgebraic Boolean
  matrix. If the description complexity of every column of $M$ does
  not exceed $\slocal$, then every probabilistically dense submatrix
  of $M$ has the $m$-majority property for $n\leq
  m=\bigO(n\log\slocal)$.
\end{lem}

\begin{proof}
  Let $M'$ be a submatrix of $M$, and assume that the matrix $M'$ is
  probabilistically dense. Since $M'$ is a submatrix of $M$, its
  growth function satisfies $\Pi_{M'}(m)\leq \Pi_M(m)$ for all $m\geq
  1$. Hence, \cref{lem:growth} gives us an upper bound
  \begin{equation}\label{eq:upperPi}
    \Pi_{M'}(m)\leq \Pi_M(m)\leq \left(\frac{8\euler m \slocal^2}{n}\right)^n\,.
  \end{equation}
  on the growth function of the matrix $M'$, for all $m\geq n$. On the
  other hand, since the matrix $M$ is semialgebraic, it is permissible
  (see \cref{app:permissible}). So, by \cref{lem:VC}, the submatrix
  $M'$ of $M$ has the $m$-majority property for any $m\geq 2/\const$
  satisfying $\Pi_{M'}(m)\leq \euler^{\const m}$, where $\const > 0$
  is an absolute constant. Thus, in order to ensure the $m$-majority
  property for the submatrix $M'$, it is enough, by \cref{eq:upperPi},
  to ensure that $m$ satisfies the inequality
  \begin{equation}\label{eq:condition}
    \left(\frac{8\euler m \slocal^2}{n}\right)^n\leq \euler^{\const m}\,.
  \end{equation}
  By taking logarithms and setting $w:=m/n$, \cref{eq:condition} turns
  into the inequality $\ln w +\ln(8\euler\slocal^2)\leq \const w$. If
  $w\leq 8\euler\slocal^2$, then it is enough that
  $2\ln(8\euler\slocal^2)\leq \const w$ holds, which happens if
  $w=C\log\slocal$ for a large enough constant $C$. If $w\geq
  8\euler\slocal^2$, then it is enough that $2\ln w\leq \const w$
  holds, which happens if $w=C$ itself is a large enough constant. In
  both cases, we have that $w\leq C\log\slocal$ and, hence, any
  integer $m\leq C n\log\slocal$ for a large enough constant $C$
  satisfies the inequality \cref{eq:condition}.
\end{proof}

\begin{rem}
  Note that in order to apply \cref{infmajrule} for a boolean matrix
  $M$, an upper bound on the description complexity $\slocal$ of its
  individual columns does not suffice. To ensure the permissibility of
  the entire matrix, we have also to ensure that the matrix itself is
  semialgebraic, that is, has an arbitrary large but \emph{finite}
  description complexity: even if the description complexity of
  individual columns is bounded, the entire matrix $M$ may be
  \emph{not} semialgebraic (see \cref{rem:columns}).  Fortunately,
  already the classical Tarski--Sei\-den\-berg
  theorem~\cite{tarski,seidenberg} (superpositions of semialgebraic
  functions are semialgebraic) ensures that graph matrices of
  probabilistic circuits consisting of semialgebraic gates \emph{are}
  semialgebraic.
\end{rem}

\section{Step 3: Description complexity of circuits}
\label{sec:quantifiers}
An important consequence of the Tarski--Sei\-den\-berg
theorem~\cite{tarski,seidenberg} is that compositions of semialgebraic
functions are also semialgebraic functions. This, in particular,
implies that functions computable by circuits over any basis
consisting of semialgebraic functions are also semialgebraic. But what
about the description complexity of such functions?
\begin{itemize}
\item[$(\ast)$] If each basis function (gate) has description
  complexity at most $\dgates$, how large can then the description
  complexity of functions computable by circuits of size up to $\size$
  be?
\end{itemize}

The answer is given in the following lemma.

\begin{lem}
  \label{lem:descrcomplexity}
  Every function $f:\RR^n\to\RR$ computable by a deterministic
  $\dgates$-semialgebraic circuit of size at most $\size$ has the
  following properties.
  \begin{enumerate}
  \item[(i)] The graph $\{(x,y)\colon f(x)=y\}$ of $f$ can be defined
    by an existential algebraic formula of description complexity at
    most $\dgates\size$, and with at most $\size-1$ (existential)
    quantifiers.

  \item[(ii)] The function $f$ is $\slocal$-semialgebraic for
    $\slocal$ satisfying $\log \slocal=\bigO(n\size\log
    \dgates\size)$.
  \end{enumerate}
\end{lem}

\begin{proof}
  The second property (ii) follows directly from (i) and
  \cref{cor:elimination}. So, it is enough to prove the first
  property~(i).

  Let $\basis$ be a basis consisting of
  $\dgates$-se\-mi\-al\-ge\-bra\-ic functions.  Let $\F$ be a circuit
  of size $\size$ over $\basis$ computing the function
  $f:\RR^n\to\RR$.  The circuit $\F$ is a sequence
  $\F=(f_1,\ldots,f_{\size})$ of functions $f_i: \RR^n\to \RR$, where
  $f_s=f$ and each $f_i$ is obtained by applying one of the basis
  operations (a gate) to
  $\RR\cup\{x_1,\ldots,x_n,f_1,\ldots,f_{i-1}\}$. Since every basis
  operation $g_i: \RR^k\to \RR$ is $\dgates$-semialgebraic, there must
  be an algebraic formula $\Phi_i(x,y)$ using at most $\dgates$
  polynomials of degree at most $\dgates$ such that $\Phi_i(x,y)=1$ if
  and only if $y=g_i(x)$.

  Replace now each gate $f_i$ in $\F$ by a new variable $z_i$. Then
  every gate $f_i=g_i(f_{i_1},\ldots, f_{i_k})$ with $g_i\in\basis$
  and each $f_{i_j}$ in $\RR\cup\{x_1,\ldots,x_n,f_1,\ldots,f_{i-1}\}$
  gives us an equation $z_i=g_i(w_i)$, where $w_i$ is a vector in
  $(\RR\cup\{x_1,\ldots,x_n,z_1,\ldots,z_{i-1}\})^k$. So,
  $\Phi_i(w_i,z_i)=1$ if and only if $z_i=g_i(w_i)$.  The value of the
  first variable $z_1$ in the sequence $z_1,\ldots,z_s$ is determined
  by the actual inputs $\RR\cup\{x_1,\ldots,x_n\}$ to the circuit (is
  obtained by applying the basis operation $g_1$ to these inputs),
  whereas the value of each subsequent variable $z_i$ ($i\geq 2$) is
  obtained by applying the $i$th base operation $g_i$ to these inputs
  and some of the previous values $z_1,\ldots,z_{i-1}$. So, the
  existential formula
  \begin{align*}
    \Psi(x,y)&= \exists z_1\ldots \exists z_{s-1}
    \pred{z_1=g_1(w_1)}\land\cdots\land
    \pred{z_{s-1}=g_{s-1}(w_{s-1})}\land
    \pred{y=g_{s}(w_{s})}\\
    &=\exists z_1\ldots \exists z_{s-1} \Phi_1(w_1,z_1)\land
    \ldots\land \Phi_{s-1}(w_{s-1},z_{s-1})\land \Phi_{s}(w_{s},y)
  \end{align*}
  defines the graph $\{(x,y)\colon y=f(x)\}$ of the function
  $f=f_{\size}$ computed by our circuit $\F$.  Existential quantifiers
  just guess the possible values taken at intermediate gates, and the
  equalities ensure their correctness.  Since each algebraic formula
  $\Phi_i$ uses at most $\dgates$ polynomials of degree at most
  $\dgates$, the formula $\Psi$ uses at most $\dgates\size$
  polynomials of degree at most $\dgates$, and contains only $\size-1$
  quantifiers.
\end{proof}

\section{Proof of Theorem~\ref{thm:main1}}
\label{sec:main1}
Suppose that a probabilistic $\dgates$-semialgebraic circuit
$\F(x,\xr)$ of size $\size$ with $k$ random input variables computes a
function $f:\RR^n\to \RR$. Our goal is to show then there are
$m=\bigO(n^2\size\log \dgates\size)$ deterministic copies
$\F_1(x,\rr_1),\ldots,\F_m(x,\rr_m)$ of $\F(x,\xr)$ such that, for
every input $x\in\RR^n$, more than the half of these circuits will
output the correct value~$f(x)$.

Let $M\colon \RR^n\times \RR^k\to\{0,1\}$ be the correctness matrix of
the circuit $\F$ (with respect to the given function $f$). Hence, the
entries of $M$ are defined by: $M[x,\rr]=1$ if and only if $\F(x,\rr)=
f(x)$.

\begin{clm}\label{clm:majprop}
  The matrix $M$ has the $m$-majority property for
  $m=\bigO(n^2\size\log \dgates\size)$.
\end{clm}

\begin{proof}
  We are going to apply the infinite majority rule
  (\cref{infmajrule}).  Recall that the graph matrix of the circuit
  $\F(x,\xr)$ is the Boolean matrix $M_F\colon \RR^{n+1}\times
  \RR^k\to\{0,1\}$ with entries defined by: $M_F[(x,y),\rr]=1$ if and
  only if $y=\F(x,\rr)$.

  Since the circuit $\F$ only uses semialgebraic functions as gates,
  Tarski--Sei\-den\-berg theorem~\cite{tarski,seidenberg} implies that
  the function $\F:\RR^n\times\RR^{k}\to\RR$ computed by the circuit
  $\F$ and, hence, also the graph matrix $M_F$ of $\F$ is also
  semialgebraic.  Furthermore, for every assignment $\rr\in\RR^n$ of
  the values to the random input variables, $\F(x,\rr)$ is a
  deterministic $\dgates$-semialgebraic circuit of size $\size$
  computing some function
  $\F_{\rr}:\RR^n\to\RR$. \Cref{lem:descrcomplexity} implies that each
  of the functions $\F_{\rr}$ is $\slocal$-semialgebraic for $\slocal$
  satisfying $\log \slocal=\bigO(n\size\log \dgates\size)$.  Thus, the
  description complexity of every column of $M_F$ does not
  exceed~$\slocal$.

  Note that the correctness matrix $M$ is a \emph{submatrix} of the
  matrix $M_F$ obtained by removing all rows of $M_F$ labeled by pairs
  $(x,y)$ such that $y\neq f(x)$, and replacing the label $(x,y)$ of
  each remaining row by~$x$.  Moreover, since the (probabilistic)
  circuit $\F(x,\xr)$ computes $f$, the correctness matrix $M$ is
  probabilistically dense. (The graph matrix $M_F$ itself does not
  need to be such.) So, the infinite majority rule (\cref{infmajrule})
  implies that the correctness matrix $M$ has the $m$-majority
  property for $m=\bigO(n\log\slocal)=\bigO(n^2\size\log
  \dgates\size)$.
\end{proof}

\Cref{clm:majprop} implies that there must be some $m$ (not
necessarily distinct) columns $\rr_1,\ldots,\rr_m$ of $M$ such that,
for every input $x\in\RR^n$, the inequality $\left|\{i\colon
  M[x,\rr_i]=1\}\right|>m/2$ and, hence, also the inequality
$\left|\{i\colon \F(x,\rr_i)=f(x)\}\right|>m/2$ holds.  Thus, on every
input $x\in\RR^n$, more than the half of the values computed by
deterministic copies $\F_1(x,\rr_1),\ldots,\F_m(x,\rr_m)$ of the
circuit $\F(x,\xr)$ compute the correct value~$f(x)$, as desired.
\qed

\begin{rem}
  We could apply \cref{lem:descrcomplexity} to the function
  $\F:\RR^n\times\RR^k\to\RR$ computed by the entire circuit
  $\F(x,\xr)$, but this would result in quadratic increase of the size
  of the derandomized circuit. Namely, \cref{lem:descrcomplexity}
  would then imply that this function is $\slocal$-semialgebraic with
  $\log \slocal$ at most about $(n+k)\size\log \dgates\size$. Since
  the number $k$ of random input variables may be as large as the
  total number $\size$ of gates, this is then about $\size^2\log
  \dgates\size$, and we would obtain by a multiplicative factor
  $\size$ worse upper bound $m=\bigO(n^2\size^2\log \dgates\size)$ in
  \cref{clm:majprop}.
\end{rem}

\section{Proof of Theorem~\ref{thm:main2}}
\label{sec:main2}
Let $x\rel y$ be a $\relcompl$-semialgebraic relation, and
$f:\RR^n\to\RR$ a $\sslocal{f}$-semialgebraic function. Suppose that
$f$ can be $\rel$-approximated by a probabilistic
$\dgates$-semialgebraic circuit $\F(x,\xr)$ of size $\size$. Our goal
is to show that then $f$ can be also $\rrel$-approximated as a
majority $\rrel$-vote of $m=\bigO(n^2\size\log \K)$ deterministic
copies of this circuit, where $\K=\dgates\size+\sslocal{f}
+\relcompl$.

Consider the \emph{correctness matrix} $M\colon \RR^n\times
\RR^k\to\{0,1\}$ with entries defined by:
\[
\mbox{$M[x,\rr]=1$ if and only if $\F(x,\rr)\rel f(x)$.}
\]
Since the circuit $\F$ $\rel$-approximates the function, the matrix
$M$ is probabilistically dense.

 \begin{clm}\label{clm:rel}
   The correctness matrix $M$ is semialgebraic, and the description
   complexity $\slocal$ of every its column satisfies
   $\log\slocal=\bigO(n\size\log\K)$.
 \end{clm}

\begin{proof}
  The probabilistic circuit $\F(x,\xr)$ computes some function
  $F:\RR^n\times\RR^k\to\RR$ of $n+k$ variables. Let $\Phi_F(x,y,\rr)$
  an existential algebraic formula ensured by
  \cref{lem:descrcomplexity}(i). Hence, the formula $\Phi_F$ has at
  most $\size-1$ quantifiers, has description complexity $\kappa\leq
  \dgates\size$, and defines the graph of, that is,
  $\Phi_F(x,y,\rr)=1$ if and only if $y=\F(x,\rr)$.

  Similarly, since the function $f$ is $\sslocal{f}$-semialgebraic,
  there is an algebraic formula $\Phi_{f}(x,y)$ of size and degree at
  most $\sslocal{f}$ such that $\Phi_{f}(x,y)=1$ if and only if
  $y=f(x)$. Finally, since the relation $\rrel$ is
  $\relcompl$-semialgebraic, there is an algebraic formula
  $\Phi_{\rrel}(x,y)$ of size and degree at most $\relcompl$ such that
  $\Phi_{\rrel}(x,y)=1$ if and only if $x\rel y$.  Consider the
  existential algebraic formula
  \[
  \Psi(x,\rr) = \exists y_1\ \exists y_2\ \Phi_{F}(x,y_1,\rr)\ \land\
  \Phi_{f}(x,y_2)\ \land \ \Phi_{\rrel}(y_1,y_2)\,.
  \]
  It is easy to see that for every row $x\in\RR^n$ and every column
  $\rr\in\RR^k$ of $M$, we have $M[x,\rr]=1$ if and only if
  $\Psi(x,\rr)=1$. Indeed, since both $\F(x,\rr)$ and $f(x)$ are
  everywhere defined functions, on every point $(x,\rr)$ they output
  some unique values $\F(x,\rr)=y_1$ and $f(x)=y_2$. So, the first
  part $\exists y_1\ \exists y_2\ \Phi_{F}(x,y_1,\rr)\land
  \Phi_{f}(x,y_2)$ of the formula $\Psi$ is a tautology, that is,
  outputs $1$ on all inputs. But the last formula
  $\Phi_{\rrel}(y_1,y_2)$ outputs $1$ precisely when $y_1\rel y_2$
  holds, which happens precisely when $\F(x,\rr)\rel f(x)$ holds.

  Thus, the existential formula $\Psi(x,\rr)$ defines the correctness
  matrix $M$.  By the Tarski--Sei\-den\-berg
  theorem~\cite{tarski,seidenberg}, the formula $\Psi(x,\rr)$ has an
  equivalent quantifier-free algebraic formula. This shows that the
  correctness matrix $M$ is semialgebraic, and it remains to upper
  bound the description complexity of its columns.

  So, fix a column $\rr\in\RR^k$ of $M$, and consider the existential
  formula $\Psi_{\rr}(x):= \Psi(x,\rr)$ obtained from the formula
  $\Psi(x,\rr)$ by fixing the $\rr$-variables to the corresponding
  values. This formula defines the $\rr$th column of $M$, and its
  description complexity is at most the sum
  $\kappa+\sslocal{f}+\relcompl\leq
  \dgates\size+\sslocal{f}+\relcompl\leq \K$ of the description
  complexities of formulas $\Phi_{F}$, $\Phi_{f}$ and $\Phi_{\rrel}$.
  The formula $\Psi_{\rr}$ has $n$ free variables ($x$-variables). The
  formulas $\Phi_{f}$ and $\Phi_{\rrel}$ have no quantifiers, and $
  \Phi_{F}$ has at most $\size$ existential quantifiers. So, the
  entire existential formula $\Psi_{\rr}$ has only $q\leq s+2$
  quantifiers, and its description complexity is at most $\K$.
  \Cref{cor:elimination} gives us an equivalent quantifier-free
  algebraic formula of description complexity $\slocal$ satisfying
  $\log\slocal=\bigO(nq\log \K)=\bigO(ns\log\K)$.  Thus, the
  description complexity $\slocal$ of each single column of $M$
  satisfies $\log\slocal=\bigO(n\size\log\K)$, as desired.
\end{proof}

Since the circuit $\F(x,\xr)$ $\rrel$-approximates our function $f$,
the correctness matrix $M$ is probabilistically dense.  By
\cref{clm:rel}, the description complexity $\slocal$ of every its
column satisfies $\log\slocal=\bigO(n\size\log\K)$.  So, by the
infinite majority rule (\cref{infmajrule}), the matrix $M$ has the
$m$-majority property for
$m=\bigO(n\log\slocal)=\bigO(n^2\size\log\K)$.  This means that there
must be some $m$ (not necessarily distinct) columns
$\rr_1,\ldots,\rr_m$ of $M$ such that, for every input $x\in\RR^n$,
the inequality $\left|\{i\colon M[x,\rr_i]=1\}\right|>m/2$ and, hence,
also the inequality $\left|\{i\colon \F(x,\rr_i)\rel
  f(x)\}\right|>m/2$ holds.  Thus, if $\Maj:\RR^m\to\RR$ is a majority
$\rrel$-vote function, then $ \Maj(\F_1(x,\rr_1),\ldots,\F_m(x,\rr_m))
\rel f(x) $ holds for every input $x\in\RR^n$. That is, the obtained
deterministic circuit (with one majority $\rrel$-vote output gate)
$\rrel$-approximates the values $f(x)$ of our function $f$, as
desired.  \qed

\subsection{Circuits approximating optimization problems}
\label{sec:contig}

Since one of the motivations in this paper is to derandomize
probabilistic dynamic programming algorithms, let us demonstrate
\cref{thm:main2} on semialgebraic circuits solving optimization
problems. The minimization problem $f:\RR^n\to\RR$ on a finite set
$A\subset\NN^n$ of feasible solutions is to compute the values
$f(x)=\min\left\{a_1x_1+\cdots a_nx_n \colon a\in A\right\}$ on all
input weighings $x\in\RR^n$.

A probabilistic circuit $\F(x,\xr)$ approximates the problem $f$
within a given factor $\factor\geq 0$ if for every input weighting
$x\in\RR^n$, $|\F(x,\xr)-f(x)|\leq \factor$ holds with probability at
least~$2/3$.

The relation $\rel$ is this case is: $x\rel y$ if and only if
$|x-y|\leq \factor$ (the fourth relation in \cref{ex:relations}).
This relation can be defined by a trivial algebraic formula
$\pred{x\geq y-\factor}\land\pred{x\leq y+\factor}$. The formula uses
only two polynomials $x-y-\factor$ and $x-y+\factor$ of degree $1$;
so, the description complexity is $\relcompl\leq 2$. The relation is
clearly contiguous: if $x\leq y\leq z$, $|x-a|\leq \factor$ and
$|z-a|\leq \factor$, then also $|y-a|\leq \factor$.

Let $\basis$ be any basis containing the optimization operations
$\min(x,y)$, $\max(x,y)$ and any other operations of a constant
description complexity $\dgates=\bigO(1)$. For example, besides $\min$
and $\max$, the basis may contain any of the arithmetic operations
$+,-,\times,\div$, any branching operations ``if $x\Diam y$ then $u$
else $v$'' with $\Diam\in\{>, \geq, =,\leq, <\}$, and other
operations.

\begin{cor}\label{cor:approx}
  If a minimization problem $f(x)=\min\left\{a_1x_1+\cdots a_nx_n
    \colon a\in A\right\}$ can be approximated within some fixed
  factor by a probabilistic circuit of size $\size$ over the basis
  $\basis$, then $f$ can be also approximated within the same factor
  by a deterministic circuit over $\basis$ of size at most a constant
  times $n^2\size^2 \log(\size+|A|)$.
\end{cor}

\begin{proof}
  The graph $\{(x,y)\colon y=f(x)\}$ of the function $f$ can be
  defined by an algebraic formula
  \[
  \bigwedge_{a\in A}\pred{a_1x_1+\cdots a_nx_n-y\geq 0} \land\Big(
  \bigvee_{a\in A}\pred{a_1x_1+\cdots a_nx_n - y=0}\Big)
  \]
  using $|A|$ polynomials of degree $1$. So, the description
  complexity of $f$ is $\sslocal{f}\leq |A|$.  Since the approximation
  relation $\rrel$ in our case has a constant description complexity
  $\relcompl\leq 2$, and since the description complexity $\dgates$ of
  every gate is also constant, \cref{thm:main2} implies that the
  minimization problem $f$ can be approximated as a majority
  $\rrel$-vote function of $m=\bigO(n^2\size \log \K)$ deterministic
  copies of the probabilistic circuit, where $\K= \dgates\size +
  \sslocal{f}+\relcompl= \bigO(\size + |A|)$.

  Since the relation $\rrel$ is contiguous, and since both $\min$ and
  $\max$ operations are available, a majority $\rrel$-vote function of
  $m$ variables can be computed by a circuit over $\basis$ of size
  $\bigO(m\log m)$ (see \cref{clm:contiguous} in \cref{app:not-maj}).
  Thus, the size of the derandomized circuit is at most a constant
  times $m\cdot\size+m\log m$, which is at most a constant times
  $n^2\size^2 \log(\size+|A|)$, as desired.
\end{proof}

\begin{rem}
  Note that the upper bound on the size $S$ of the derandomized
  circuit, given by \cref{cor:approx}, is only \emph{logarithmic} in
  the number $|A|$ of feasible solutions of the minimization problem
  $f$.  In most optimization problems, the set $A$ of feasible
  solutions is the set $A\subseteq\{0,1\}^n$ of characteristic $0$-$1$
  vectors of objects of interest: spanning trees, perfect matchings,
  etc. In these cases, $\log|A|$ is at most the number $n$ of
  variables. Thus, for such problems $f$, the size of the derandomized
  circuit is at most a constant times $n^3\size^2 \log\size$.
\end{rem}

\section{Derandomization via isolating sets}
\label{sec:isol}

\Cref{thm:main1,thm:main2} derandomize very general classes of
probabilistic circuits, but their proofs rely on deep tools from three
different fields: combinatorial algebraic geometry (sign patterns of
polynomials), probability theory (uniform convergence in probability),
and quantifier elimination theory over the reals. When directly
applied, elementary tools like the finite majority rule
(\cref{clm:fin-maj}) fail for such circuits already because the domain
is infinite.

In some cases, however, it is still possible to apply even such
elementary tools also for circuits working over infinite domains. In
particular, this happens if the functions computed by a given class of
circuits have finite ``isolating sets.'' In this section, we will
demonstrate this approach on arithmetic and tropical circuits.

Given a family $\h$ of functions $h:\dom\to\range$ and a function
$f:\dom\to\range$, a set $X\subseteq \dom$ \emph{isolates} the
function $f$ \emph{within} $\h$ if for every function $h\in\h$,
\[
\mbox{$h(x)=f(x)$ for all $x\in X$ implies that $h(x)=f(x)$ holds for
  all $x\in\dom$.}
\]
That is, if $h(x)\neq f(x)$ holds for some point $x\in\dom$ of the
entire domain $\dom$, then also $h(x)\neq f(x)$ holds for at least one
point $x\in X$.

\subsection{Arithmetic circuits}
\label{sec:arithm}
In the case of (arithmetic) polynomials, we have the following strong
isolation property.

\begin{lem}\label{lem:SZ}
  Let $f$ be a nonzero $n$-variate polynomial of degree $d$ over
  $\RR$, and $S\subset\RR$ a finite subset of $|S|\geq d+1$
  elements. Then every subset $X\subseteq S^n$ of size
  $|X|>d|S|^{n-1}$ isolates $f$ within all polynomials of degree at
  most~$d$.
\end{lem}

\begin{proof}
  Let $S\subset\RR$ a finite subset of $|S|\geq d+1$ elements. Take an
  arbitrary $n$-variate polynomial $g(x)$ of degree at most $d$, and
  suppose that $g(a)\neq f(a)$ holds for at least one point
  $a\in\RR^n$. Then $p(x):=f(x)-g(x)$ is a nonzero polynomial of
  degree at most $d$.  By the Schwartz--Zippel
  lemma~\cite{schwartz,zippel}, we then have $|\{a\in S^n\colon
  p(a)=0\}|\leq d|S|^{n-1}$. So, since $|X|>d|S|^{n-1}$, $p(a)\neq 0$
  must hold for at least one point $a\in X$, as desired.
\end{proof}

\begin{thm}\label{thm:arithm}
  If a rational function $f:\RR^n\to\RR$ can be computed by a
  probabilistic arithmetic $(+,\times,-,/)$ circuit, then $f$ can be
  also computed by a deterministic arithmetic $(+,\times,-,/)$ circuit
  of the same size.
\end{thm}

\begin{proof}
  The function $f$ is of the form $f(x)=p(x)/q(x)$ for some
  polynomials $p$ and $q$.  Suppose that $f$ can be computed by a
  probabilistic arithmetic $(+,\times,-,/)$ circuit $\F(x,\xr)$ of
  size $\size$. Set $d:=r+2^{\size}$, where $r$ is the maximum degree
  of $p$ and $q$. Take an arbitrary subset $S\subseteq \RR$ of size
  $|S|\geq 2d$. By the finite majority rule (\cref{clm:fin-maj}),
  there is an assignment $r\in\RR^k$ to the random input variables,
  and a subset $X\subset S^n$ of size $|X|>\tfrac{1}{2}|S|^n$ such
  that the deterministic copy $F_r(x)=\F(x,r)$ of the probabilistic
  circuit $\F(x,\xr)$ computes $f$ correctly on all inputs from
  $X$. The (deterministic) circuit $\F_r$ computes some rational
  function $F_r(x)=P(x)/Q(x)$.  Since the gates have fanin two, the
  polynomials $P$ and $Q$ have degrees at most $2^{\size}$. Consider
  the polynomial $g(x):=p(x)\cdot Q(x)-q(x)\cdot P(x)$. By the choice
  of $d$, the degree of the polynomial $g$ is at most $d$. We have
  only to show that $g$ is a zero polynomial, i.e., that $g(x)=0$
  holds for all $x\in\RR^n$.

  Were $g$ a nonzero polynomial, then Lemma~\ref{lem:SZ} would require
  the set $X$ to have cardinality $|X|\leq d|S|^{n-1}$.  But then we
  would have $\tfrac{1}{2}|S|^n < |X|\leq d|S|^{n-1}$ and, hence, also
  $|S|< 2d$, which contradicts our choice of $S$.
\end{proof}

\subsection{Tropical circuits}

We now consider circuits over the tropical semiring $(\RR_+,\max,+)$.
Since the basis operations $\max(x,y)$ and $x+y$ of such circuits have
very small (constant) description complexities, \cref{thm:main1}
implies that if an maximization problem $f:\RR^n\to\RR$ can be solved
by a probabilistic tropical \maxplus circuit of size $\size$, then $f$
can be also solved as a majority vote of about $n^2\size\log \size$
deterministic copies of this circuits.

But tropical \maxplus and \minplus circuits \emph{cannot} compute the
majority vote function at all (see \cref{clm:maj-impos} in
\cref{app:not-maj}).  So, the resulting deterministic circuits are
\emph{not} \maxplus circuits.

On the other hand, tropical circuits are interesting in optimization,
because they simulate so-called \emph{pure} dynamic programming
algorithms (pure DP algorithms). This raises the question: can
probabilistic pure DP algorithms be (efficiently) derandomized at
least in the one-sided error probability scenario? In this section, we
will give an \emph{affirmative} answer: under the \emph{one-sided}
error probability scenario, the resulting deterministic circuits are
also tropical circuits (do not use majority vote gates), and the
derandomization itself is then elementary.

What circuits over the arithmetic semiring $(\RR_+,+,\times)$ compute
are polynomials
\begin{equation}\label{eq:pol-arith}
  p(x)=\sum_{a\in A}c_a\prod_{i=1}^n x_i^{a_i}\,,
\end{equation}
where $A\subset\NN^n$ is some finite set of nonnegative integer
exponent vectors, and $c_a\in\RR_+$ are positive coefficients. A
polynomial of the form \cref{eq:pol-arith} is \emph{multilinear} if
the degree of every variable is at most $1$, that is, if $A\subseteq
\{0,1\}^n$. We also call a polynomial of the form \cref{eq:pol-arith}
\emph{constant-free} if it has no nonzero coefficients different from
$1$, that is, if $c_a=1$ for all $a\in A$.

In the tropical semiring $(\RR_+,\max,+)$, the arithmetic addition
$x+y$ turns into $\max(x,y)$, and the arithmetic multiplication
$x\times y$ turns into $x+y$. So, what a tropical \maxplus circuit
computes is a tropical polynomial
\begin{equation}\label{eq:pol-trop}
  f(x)=\max_{a\in A}\ \skal{a,x}+c_a\,,
\end{equation}
where $\skal{a,x}=a_1x_1+\cdots+a_nx_n$ stands for the scalar product
of vectors $a$ and~$x$. That is, \maxplus circuits solve maximization
problems with linear objective functions; the set $A$ is then the set
of feasible solutions.  By analogy with arithmetic polynomials, we
call a tropical polynomial \cref{eq:pol-trop} \emph{constant-free} if
$c_a=0$ holds for all $a\in A$, and \emph{multilinear} if
$A\subseteq\{0,1\}^n$. (Note that, in the tropical semiring
$(\RR_+,\max,+)$, the multiplicative unity element ``$1$'' is $0$,
because $x+0=0+x=x$.)

Under a \emph{probabilistic \maxplus circuit} of size $\size$ we will
now understand an \emph{arbitrary} random variable $\xf$ taking its
values in the set of all deterministic \maxplus circuits of size at
most~$\size$. That is, we now do not insist that the randomness into
the circuits can be only introduced via random input variables.  Such
a circuit $\xf$ \emph{solves} a given maximization problem $f:R^n\to
R$ with \emph{one-sided} success probability $0\leq p\leq 1$ if for
every input weighting $x\in\RR_+^n$, we have
\[
\mbox{$\prob{\xf(x) > f(x)}=0$ and $\prob{\xf(x) < f(x)}\leq 1-p$.}
\]
That is, the circuit is not allowed to output any better than
``optimum'' value $f(x)$, but is allowed to output worse values with
probability at most $1-p$. In particular, $p=1$ means that the circuit
must correctly compute $f$, while $p=0$ means that the circuit can do
``almost everything,'' it only must never output better than optimal
values.

As in \cref{sec:arithm}, we will use the approach of isolating
sets. Let $f$ be an $n$-variate \maxplus polynomial. A set
$X\subseteq\RR_+^n$ of input weights is \emph{isolating} for $f$ if
for every $n$-variate \maxplus polynomial~$h$,
\[
\mbox{$h(x)=f(x)$ for all $x\in X$ implies that $h(x)=f(x)$ holds for
  all $x\in\RR_+^n$.}
\]

In the case of tropical polynomials, we do not have such a strong
isolation fact as \cref{lem:SZ}.  Still, also then some specific sets
of input weighings are isolating. In the case of \maxplus polynomials,
such is the set of all $0$-$1$ weighings.

\begin{lem}\label{lem:isol1}
  Let $f$ be a \maxplus polynomial of $n$ variables. If $f$ is
  multilinear and constant-free, then the set $X=\{0,1\}^n$ is
  isolating for~$f$.
\end{lem}

\begin{proof}
  Let $f(x)=\max_{a\in A} \skal{a,x}+c_a$ be a \maxplus
  polynomial. Since $f$ is multilinear, we have $A\subseteq\{0,1\}^n$,
  and since $f$ is constant-free, we also have $c_a=0$ for all $a\in
  A$. Now take an arbitrary \maxplus polynomial $h(x)=\max_{b\in B}
  \skal{b,x}+c_b$, and suppose that \begin{equation}\label{eq:isol}
    \mbox{$h(x)=f(x)$ holds for all input weightings $x\in\{0,1\}^n$.}
  \end{equation}
  Our goal is to show that then $h(x)=f(x)$ holds for \emph{all}
  nonnegative real weighings $x\in\RR_+^n$.

  Since the polynomial $f$ is constant-free, $f(\vec{0})=0$ holds for
  the all-$0$ input weighting $\vec{0}$. Together with \cref{eq:isol},
  this yields $h(\vec{0})=0$.  Since the ``coefficients''
  $c_b\in\RR_+$ of the polynomial are \emph{nonnegative}, and since
  the polynomial $h$ takes the \emph{maximum} of the values $
  \skal{b,x}+c_b$, the equality $h(\vec{0})=0$ implies $c_b=0$ for all
  $b\in B$. So, both polynomials $f$ and $h$ are constant-free.

  Furthermore, since $g(x)=f(x)$ must hold for each of $n$ input
  weighings $x\in\{0,1\}^n$ with exactly one $1$, all vectors in $B$
  must also be $0$-$1$ vectors. The vectors $a$ in $A$ and $B$ can
  therefore be identified with their supports $\supp{a}=\{i\colon
  a_i=1\}$. We claim that:
  \begin{enumerate}
  \item[(i)] the support of every vector of $B$ lies in the support of
    at least one vector of $A$, and
  \item[(ii)] the support of every vector of $A$ lies in the support
    of at least one vector of $B$.
  \end{enumerate}
  To show (i), suppose contrariwise that there is a vector $b\in B$
  such that $\supp{b}\setminus\supp{a}\neq\emptyset$ holds for all
  $a\in A$. Then on the $0$-$1$ input $x=b\in \{0,1\}^n$, we have
  $g(x)\geq \skal{b,x}=\skal{b,b}=|\supp{b}|$. But since every vector
  $a\in A$ has a zero in some position $i\in\supp{b}$, we have
  $\skal{a,x}=\skal{a,b}\leq |\supp{b}|-1$ and, hence, also $f(x)\leq
  |\supp{b}|-1$, a contradiction with \cref{eq:isol}. The argument for
  the property (ii) is the same with the roles of $A$ and $B$
  interchanged.

  Now, for every input weighting $x\in\RR_+^n$, property (i) gives the
  inequality $h(x)\leq f(x)$, while (ii) gives the converse
  inequality.
\end{proof}

\begin{thm}\label{thm:max}
  If a multilinear and constant-free \maxplus polynomial $f$ can be
  computed by a probabilistic \maxplus circuit of size $s$ with
  one-sided success probability $p>0$ then $f$ can be also computed by
  a deterministic \maxplus circuit of size at most $(s+1)\lceil n/p
  \rceil$.
\end{thm}

Note that the size of the obtained deterministic circuits remains
proportional to~$ns$ even if the success probability $p>0$ is an
arbitrarily small constant.  This is in sharp contrast with the
two-sided error scenario, where we required the success probability to
be $p\geq 1/2+c$ for a constant~$c>0$ (for definiteness, we have used
$p=2/3$).

\begin{proof}
  By \cref{lem:isol1}, we know that the set $X=\{0,1\}^n$ isolates $f$
  within all \maxplus polynomials.  Let $\xf$ be a probabilistic
  \maxplus circuit of size $s$ computing $f$ with a one-sided success
  probability $p>0$. Take $m=\lceil(1/p) \log|X|\rceil=\lceil n/p
  \rceil$ independent copies $\xf_1,\ldots,\xf_m$ of the circuit
  $\xf$, and consider the probabilistic \maxplus circuit
  $\xh(x)=\max\left\{\xf_1(x), \ldots, \xf_m(x)\right\}$.

  Fix a vector $x\in X$.  Since only \emph{one-sided} error $\err=1-p$
  is allowed, we know that $\xf_i(x)\leq f(x)$ must hold for all $i$.
  Hence, $\xh(x)\neq f(x)$ can only happen when \emph{all} the values
  $\xf_1(x),\cdots,\xf_m(x)$ are strictly smaller than the optimal
  value $f(x)$, and this can only happen with probability at most
  $\err^m=(1-p)^m\leq \euler^{-pm}$.  So, by the union bound, the
  probability that $\xh(x)\neq f(x)$ holds for at least one of the
  inputs $x\in X$ does not exceed $|X|\err^m\leq |X| \euler^{-pm}$,
  which is smaller than $1$, because $m\geq (1/p)\log|X|$ (and
  $\log\euler >1$).

  There must therefore be a realization $H(x)=\max\left\{F_1(x),
    \ldots, F_m(x)\right\}$ of the probabilistic circuit $\xh$ such
  that the polynomial $h(x)$ computed by $H(x)$ satisfies $h(x)=f(x)$
  for all $x\in X$. The size of the obtained deterministic circuit
  $H(x)$ is at most $ms+m-1\leq (s+1)\lceil n/p \rceil$.  Since the
  set $X$ is isolating for $f$, the fact that $h(x)=f(x)$ holds for
  all $x\in X$ implies this implies $h(x)=f(x)$ holds for all $x\in
  \RR_+^n$, that is, the obtained deterministic circuit $H$ correctly
  computes $f$ on all possible inputs.
\end{proof}

\subsection*{Acknowledgements}
  I am thankful to Joshua Grochow, Pascal Koiran, Igor Sergeev and
  Hans Ulrich Simon for inspiring discussions at the initial stages of
  this investigation.  This work is supported by the
  German Research Foundation (DFG) under Grant
  JU~3105/1-1.

\appendix

\section{Circuits for majority vote}
\label{app:not-maj}

Recall that the \emph{majority vote} function of $m$ variables is a
partly defined function $\maj_n(x_1,\ldots,x_n)$ that outputs the
majority element of its input string $x_1,\ldots,x_n$, if there is
one.

\begin{clm}\label{clm:maj-impos}
  Arithmetic $(+,-,\times)$ circuits, as well as tropical \minplus and
  \maxplus circuits cannot compute majority vote functions.
\end{clm}

\begin{proof}
  Functions computed by circuits over the arithmetic basis
  $\{+,-,\times\}$ are polynomial functions. So, suppose contrariwise
  that we can express $\maj(x,y,z)$ as a polynomial
  $f(x,y,z)=ax+by+cz+h(x,y,z)$, where the polynomial $h$ is either a
  zero polynomial or has degree $>1$. Then $f(x,x,z)=x$ implies $c=0$,
  $f(x,y,x)=x$ implies $b=0$, and $f(x,y,y)=y$ implies $a=0$. This
  holds because, over fields of zero characteristic, equality of
  polynomial functions means equality of coefficients.  We have thus
  shown that $h=\maj$. So, the polynomial $h$ cannot be the zero
  polynomial. But then $h$ has degree $> 1$, so $h(x,x,x)=x$ for all
  $x\in\RR$ is impossible.

  Let us now show that also tropical circuits cannot compute majority
  vote functions. Every tropical \minplus circuit computes some
  tropical \minplus polynomial. The functions $f:\RR^n\to\RR$ computed
  by tropical \minplus polynomials are piecewise linear \emph{concave}
  functions. In particular, $f(\tfrac{1}{2}x+\tfrac{1}{2}y)\geq
  \tfrac{1}{2}f(x)+\tfrac{1}{2}f(y)$ must hold for all $x,y\in\RR^n$:
  \[
  \min_{v\in V}\ \skal{v,x+y}\geq \min_{v\in V}\ \skal{v,x}+\min_{v\in
    V}\ \skal{v,y}\,.
  \]
  But already the majority vote function $\maj:\RR^3\to\RR$ of three
  variables is not concave.  To see this, take two input vectors
  $x=(a,a,c)$ and $y=(a,b,b)$ with $a <b$ and $c=2a-b$. Then
  $\maj(\tfrac{1}{2}x+\tfrac{1}{2}y)=\maj(a,(a+b)/2,a)=a$ but
  $\tfrac{1}{2}\maj(x)+\tfrac{1}{2}\maj(y)=\tfrac{1}{2}a+\tfrac{1}{2}b
  > a$ since $b>a$.  So, $\maj$ is not concave. Similar argument shows
  that $\maj$ is not \emph{convex} and, hence, cannot be computed by
  tropical \maxplus circuits.
\end{proof}

Recall that a binary relation $\rel \subseteq\RR\times\RR$
\emph{contiguous} if $x\leq y\leq z$, $x\rel a$ and $z\rel a$ imply
$y\rel a$. That is, if the endpoints of an interval are close to $a$,
then also all numbers in the interval are close to~$a$.

\begin{clm}\label{clm:contiguous}
  For every contiguous relation $x\rel y$, a majority $\rrel$-vote
  function of $m$ variables can be computed by a fanin-$2$
  $(\min,\max)$ circuit of size $\bigO(m\log m)$.
\end{clm}

\begin{proof}
  Given a sequence $x_1,\ldots,x_m$ of real numbers, the \emph{median
    function} outputs the middle number $x_{i_{\lceil m/2\rceil}}$ of
  the sorted sequence $x_{i_1}\leq \ldots\leq x_{i_m}$.  So, the
  sorting network of Ajtai, Koml\'os and Szemer\'edi~\cite{AKS}
  computes the median function using only $\bigO(m\log m)$ $\min$ and
  $\max$ operations.  On the other hand, it is easy to see that the
  median function is a majority $\rrel$-vote function for every
  contiguous relation $x\rel y$.

  Indeed, let $x_1\leq \ldots\leq x_m$ be a sorted sequence of real
  numbers, and $a$ a real number. Call a position $i$ \emph{good}, if
  $x_i\rel a$ holds. Suppose that more than half of the positions $i$
  are good.  Since the relation $\rel$ is contiguous, good positions
  constitute a contiguous \emph{interval} of length $>m/2$. So, the
  median of $x_1,\ldots,x_m$ must be the number $x_i$ in a good
  position~$i$.
\end{proof}

Recall that the \emph{nullity relation} $x\rel y$ holds precisely when
either both $x=0$ and $y=0$, or both $x\neq 0$ and $y\neq 0$ hold.  A
\emph{zero vote function} of $n$ variables is any function
$f:\RR^n\to\RR$ such that $f(x_1,\ldots,x_n)=0$ precisely when more
than $n/2$ of the numbers $x_i$ are zeros. Note that every zero-vote
function is a majority $\rrel$-vote function for the nullity
relation~$\rel$: either more than half of all numbers $x_1,\ldots,x_n$
are zeros, or more than half of them are nonzero.

\begin{clm}\label{clm:nullity}
  A zero-vote function of $n$ variables can be computed a
  $(\min,\max,\times)$ circuit of size $\bigO(n\log n)$, as well as by
  a monotone fanin-$2$ arithmetic $(+,\times)$ circuit of size
  $\bigO(n^2)$.
\end{clm}

\begin{proof}
  First, suppose that we have $(\min,\max,\times)$ among the basis
  operations. Then we can just sort the sequence $x_1^2,\ldots,x_n^2$
  of squares using $\bigO(n\log n)$ $(\min,\max)$ gates, and output
  the median of the sorted sequence $y_1\leq \ldots\leq y_n$.  Since
  the squared sequence has only nonnegative numbers, zeros (if any)
  will lie at the beginning of the sorted sequence.

  In the case of arithmetic $(+,\times)$ circuits, we can use the
  standard dynamic programming. We have only to show how to
  efficiently compute polynomials $\ssymm{m,k}{}$ such that
  $\symm{m,k}{x_1,\ldots,x_m}=0$ precisely when at least $k$ of the
  numbers $x_1,\ldots,x_m$ are zeros.  For the base cases, we can take
  $\symm{m,1}{x_1,\ldots,x_m}=x_1^2\cdots x_m^2$,
  $\symm{m,m}{x_1,\ldots,x_m}=x_1^2+\cdots +x_m^2$, and
  $\symm{m,k}{x_1,\ldots,x_m}=1$ $(\neq 0)$ for $k>m$. (We take
  squares just to avoid possible cancelations.)  Then we can use the
  recursion
  \[
  \symm{m,k}{x_1,\ldots,x_m}=\symm{m-1,k}{x_1,\ldots,x_{m-1}}\cdot
  \left[\symm{m-1,k-1}{x_1,\ldots,x_{m-1}} + x_m^2 \right]\,.
  \]
  The first polynomial $\ssymm{m-1,k}$ in this product is $0$ iff
  there are at least $k$ zeros already among the first $m-1$
  positions, whereas the second term is $0$ iff there are at least
  $k-1$ zeros among the first $m-1$ positions, and the last position
  is also zero. For $m=n$ and $k=\lfloor n/2\rfloor+1$, the obtained
  arithmetic $(+,\times)$ circuit has size $\bigO(kn)=\bigO(n^2)$, and
  computes the zero vote function.
\end{proof}

\section{On permissibility}
\label{app:permissible}
In the uniform convergence result of Vapnik and Chervonenkis given in
\cref{thm:VC}, the class of functions $H$ is required to be
permissible (see \cref{rem:permis}). While every \emph{countable}
class $H$ is permissible, uncountable classes need not automatically
be such.

Haussler in~\cite[Appendix~9.2]{haussler} gives a sufficient condition
for a class $H$ of (not necessarily $0$-$1$ valued) functions
$h:X\to\RR$ to be permissible. He calls a class $H$ \emph{indexed} by
a set $T$ if there is a real valued function $\indf$ on $T\times X$
such that $H=\{\indf(t,\cdot)\colon t\in T\}$, where $\indf(t,\cdot)$
denotes the real-valued function on $X$ obtained from $\indf$ by
fixing the first parameter to~$t$.  Haussler shows that the following
conditions are already sufficient for the class $H$ to be permissible:
(1) every function $h\in H$ is measurable, (2) the class $H$ can be
indexed by a set $T=\RR^n$ for a finite $n\geq 1$, and (3) the
indexing function $\indf:T\times X\to\RR$ itself is measurable.

In the case of Boolean semialgebraic matrices $M:T\times X\to\{0,1\}$,
we have a class $H$ of $0$-$1$ functions $h_t:X\to\{0,1\}$, where
$X=\RR^k$ and $h_t(x)=M[t,x]$. The class $H$ is indexed by the set $T$
of the form $T=\RR^n$, and the indexing function $\indf=M$ is the
matrix $M$ itself. Since the matrix $M$ is semialgebraic, the
functions $h_t\in H$ as well as the indexing function $\indf$ are
semialgebraic. Since the functions $h_t$ and the indexing function
$\indf$ are $0$-$1$ valued functions, this implies that all these
functions are measurable.

Indeed, every semialgebraic set $S\subseteq\RR^n$ is a \emph{finite}
union of \emph{finite} intersections of sets of the form
$\{x\in\RR^n\colon p(x)=0\}$ and $\{x\in\RR^n\colon p(x)>0\}$, where
$p$ is a polynomial.  Recall that a function $h:X\to \RR$ is
measurable if the set $X$ itself is a measurable set, and for each
real number~$r$, the set $S_r=\{x\in X\colon h(x) > r\}$ is
measurable. In our case, functions $h:X\to \{0,1\}$ are $0$-$1$ valued
functions. Each such function is the characteristic function of the
set $S=\{x\in X\colon h(x)=1\}$.  Then each set $S_r$ is either
$\emptyset$, $S$ or $X$.  Hence, a $0$-$1$ valued function $h$ is
measurable if and only if the set $S=h^{-1}(1)$ it represents is
measurable. Since semialgebraic sets are measurable, we have that
every semialgebraic $0$-$1$ valued function is measurable.

The books of Dudley~\cite[Chapter~10]{dudley} and
Pollard~\cite[Appendix~C]{pollard} discuss more general sufficient
conditions for classes of not necessarily $0$-$1$ valued functions
$h:X\to\RR$ to be permissible.

\end{document}